\documentclass[a4paper,aps,floats,intlimits,twocolumn,pre,10pt]{revtex4-1}
\usepackage{amsmath,amsthm,amssymb,amsthm,dsfont,mathrsfs,bbm}
\usepackage{mathtools}
\usepackage{multirow}
\usepackage{subcaption}
\usepackage[english]{babel}
\usepackage{hyperref}
\hypersetup{pdfauthor={Lucibello Parisi Sicuro},pdftitle={Polygonal corrections},%
            colorlinks, linktocpage=true, pdfstartpage=3, pdfstartview=FitV,%
    breaklinks=true, pdfpagemode=UseNone, pageanchor=true, pdfpagemode=UseOutlines,%
    plainpages=false, bookmarksnumbered, bookmarksopen=true, bookmarksopenlevel=1,%
    hypertexnames=true, pdfhighlight=/O,%
    urlcolor=orange, linkcolor=blue, citecolor=red, 
        }

\usepackage{booktabs}
\usepackage{braket}
\usepackage{pgfplots}

\pgfplotsset{compat=1.5}\usepackage{lmodern}
\usepackage[cal=boondoxo]{mathalfa}
\usepackage[T1]{fontenc}

\newcommand{\sE}{\mathcal{E}}

\newcommand{\Bx}{\mathbf{x}}

\newcommand{\Ba}{\mathbf{a}}
\newcommand{\Bb}{\mathbf{b}}
\newcommand{\Bk}{\boldsymbol{\kappa}}

\DeclareMathOperator{\atanh}{atanh}

\newcommand{\fG}{\mathtt{G}}
\newcommand{\fm}{\mathsf{m}}

\newcommand{\fK}{\mathtt{K}}
\newcommand{\fg}{\mathtt{g}}
\newcommand{\fL}{\mathtt{L}}
\newcommand{\fp}{\mathtt{p}}

\newcommand{\R}{\mathds{R}}
\newcommand{\Z}{\mathds{Z}}

\newcommand{\abs}[1]{{\left|#1\right|}}
\newcommand{\norm}[1]{{\left\|#1\right\|}}

\newcommand{\tr}[2][]{\,\mathrm{tr}_{#1}\left[{#2}\right]}

\newcommand*\pFqskip{8mu}
\catcode`,\active
\newcommand*\pFq{\begingroup
        \catcode`\,\active
        \def ,{\mskip\pFqskip\relax}%
        \dopFq
}
\catcode`\,12
\def\dopFq#1#2#3#4#5{%
        {}_{#1}F_{#2}\biggl[\genfrac..{0pt}{}{#3}{#4};#5\biggr]%
        \endgroup
}

\theoremstyle{plain}
\newtheorem*{proposition}{Proposition}

\begin{document}
\title{One-loop diagrams in the Random Euclidean Matching Problem}
\author{Carlo Lucibello}\email{carlo.lucibello@polito.it}
\affiliation{Politecnico di Torino, Corso Duca degli Abruzzi, 24, I-10129 Torino, Italy, and Human Genetics Foundation --  Torino, Via Nizza 52, I-10126 Torino, Italy}
\author{Giorgio Parisi}
\affiliation{Dipartimento di Fisica, INFN -- Sezione di Roma1, CNR-IPCF UOS Roma Kerberos, Universit\`a ``Sapienza'', P.le A. Moro 2, I-00185, Rome, Italy}
\author{Gabriele Sicuro}\email{sicuro.gabriele@for.unipi.it}
\affiliation{Centro Brasileiro de Pesquisas F\'isicas, Rua Xavier Sigaud 150, 22290--180 Rio de Janeiro -- RJ, Brazil}
\begin{abstract}
The matching problem is a notorious combinatorial optimization problem that has attracted for many years the attention of the statistical physics community. 
Here we analyze the Euclidean version of the problem, i.e. the optimal matching problem between points randomly distributed on a $d$-dimensional Euclidean space, where the cost to minimize depends on the points' pairwise distances. Using  Mayer's cluster expansion we write a formal expression for the replicated action that is suitable for a saddle point computation. We  give the diagrammatic rules for each term of the expansion, and we analyze in detail  the one-loop diagrams. A characteristic feature of the theory, when diagrams are perturbatively computed   around the mean field part of the action, is the vanishing of the mass at zero momentum. 
In the non-Euclidean case of uncorrelated costs instead, we predict and numerically verify an anomalous scaling for the sub-sub-leading correction to the asymptotic average cost.
\end{abstract}
\maketitle
\section{Introduction}
Let us consider a complete graph $\fK_N$ of $N$ vertexes, $N$ even, indexed in $[N]\coloneqq \{i\}_{i=1,\dots,N}$, and a set of cost coefficients $w_{ij}=w_{ji}$, $1\leq  i < j\leq N$, in such a way that $w_{ij}$ is associated to the (undirected) edge $(i,j)$ of the graph. The matching problem consists in finding an optimal matching $\fm$ on the graph $\fK_N$. An optimal matching $\fm$ is a subset of edges of $\fK_N$ that satisfies two fundamental properties. First, $\fm$ must be a \textit{perfect} (or admissible) matching, i.e., each vertex of $\fK_N$ must be adjacent to one, and only one, edge in $\fm$. Second, in an optimal matching $\fm$, the sum of the costs of the edges in $\fm$, also called \textit{matching cost}, is minimal (optimality condition). We can associate an occupation number $m_{ij}\in \{1,0\}$ to each edge $(i,j)$ of the original complete graph, depending on whether it belongs to the matching $\fm$, or not. We identify then the matching $\fm$ with the symmetric matrix $\mathsf m\coloneqq(m_{ij})_{ij}$, that we denote, for the sake of simplicity, by the same symbol. By means of the matrix $\fm$, the matching cost can be written as
\begin{equation}\label{costo}
\sE_w(\mathsf m)\coloneqq\sum_{i<j}^N m_{ij} w_{ij}.
\end{equation}
We can recast the original optimal matching problem into the following integer programming problem for the matrix $\fm$:
 \begin{equation}
\underset{\mathsf m}{\mathrm{minimize}}\ \sE_w(\mathsf m),
 \end{equation}
given the constraints
\begin{subequations}\label{constraints}
\begin{align}
\sum_{j=1}^N m_{ij}=1 \qquad &i=1,\dots,N,\\
m_{ij}=m_{ji},\quad \forall i,j,\quad &m_{ii}=0\quad \forall i.
\end{align}
\end{subequations}
We denote by $\fm^*$ the optimal matching, and by $\sE_w^*\coloneqq\sE_w(\fm^*)$ the optimal cost.

The study of the matching problem has a very long tradition in the literature. It is well known that, from the algorithmic point of view, the problem belongs to the $\mathrm P$ computational complexity class, as \textcite{Kuhn} proved in 1955. \textcite{Edmonds1965,Edmonds1972} later extended and improved the original result of Kuhn, showing that, for a matching problem on a generic graph $\fG$ with $V$ vertices and $E$ edges, the optimal matching can be found in $O(VE\ln E)$ iterations. Matching problems have an important theoretical relevance, but they also appear in many practical applications, such as computer vision \cite{Trucco1998}, control theory \cite{Liu2011,Menichetti2014} and pattern matching \cite{Gusfield1997} among many other fields.

Aside with the purely algorithmic aspects of the problem, however, the study of the \textit{typical} properties of the solution of a given optimization problem, respect to an ensemble of realizations, is of a certain interest. For this reason, in a set of seminal contributions, \textcite{Orland1985} and \textcite{Mezard1985,Mezard1986a,Mezard1986,Mezard1988} analyzed \textit{random} matching problems. Their works paved the way to the application of analytical tools from the theory of disordered systems to many other combinatorial optimization problems \cite{mezard2009information,Bapst2013}. In most statistical physics literature, the cost coefficients $\{w_{ij}\}_{ij}$ of a random matching problem are taken to be i.i.d.~random variables (random link approximation). The average optimal cost $\overline{\sE_w^*}$ for the random matching problem has been derived, under this assumption, in Ref.~\cite{Mezard1985}, in the large $N$ limit. The results in Ref.~\cite{Mezard1985} were later rigorously proved by \textcite{Aldous2001}. A similar analysis has been performed for the random \textit{bipartite} matching problem (or assignment problem), i.e., the random matching problem defined on a bipartite graph, in which two types of vertexes to be matched appear. A conjecture on the finite size corrections to the average optimal cost for the assignment problem was proposed in Ref.~\cite{Parisi1998} and generalized by \textcite{Coppersmith1999}. This conjecture was later proved independently by \textcite{Linusson2004} and \textcite{Nair2005}. In Ref.~\cite{Parisi2002} finite size corrections to the average optimal cost both in the random matching problem and in the random assignment problem were analyzed, using replica techniques. Many other results about random matching problems have been obtained in recent years. In particular, the theory of cavity method \cite{Mezard2003} has been successfully applied to the study of matching problems in general \cite{Altarelli2011}; an example is the evaluation of the number of solutions of the problem on sparse random graphs \cite{Zdeborova2006}, or to the study of the multi-index matching problem \cite{Martin2005}. The application of the cavity method (called belief propagation in its algorithmic version) to the assignment problem has been rigorously justified by \textcite{Bayati2008}.

In the present work, we are interested in a variation of the random matching problem in which the cost coefficients $\{w_{ij}\}_{ij}$ appearing in Eq.~\eqref{costo} are correlated random variables, due to an underlying Euclidean structure. In particular, we associate to each vertex $i$ of the complete graph a point $\Bx_i$ in the $d$-dimensional unit cube $[0,1]^d$. Then we will consider the cost coefficients to be given by
\begin{equation}
\label{w}
w_{ij}\coloneqq\norm{\Bx_i-\Bx_j}^p,\quad p>0.
\end{equation}
In the expression above, $\norm{\bullet}$ is the Euclidean norm. The points $\Bx_i$ are assumed to be independently and uniformly distributed in $[0,1]^d$, and, as usual, we are interested in the asymptotic  limit of the average (over the points' distribution) of the optimal cost. The formulated problem is therefore called (random) \textit{Euclidean} matching problem (EMP). In Fig.~\ref{fig:match2d} we present a pictorial representation of an instance of the EMP on the unit square. In the bipartite version of the EMP, or Euclidean assignment problem (EAP), two sets of points with the same cardinality are randomly generated in a certain domain, and we ask for the average optimal cost of the matching among them, requiring that points of different type only are matched. 

\begin{figure}
\includegraphics[width=0.9\columnwidth]{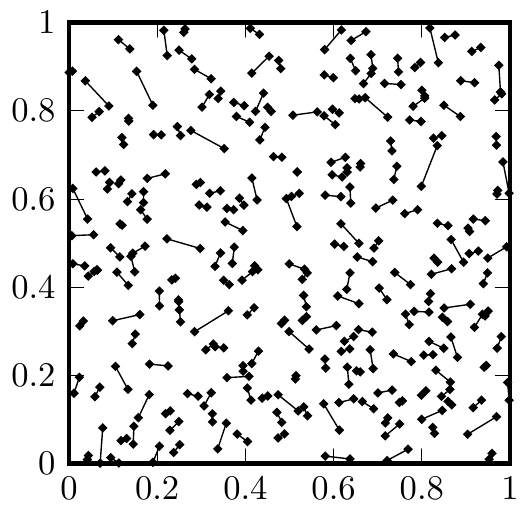}
\caption{An example of optimal matching among uniformly distributed points on the unit square, assuming open boundary conditions. The cost matrix is given by the Euclidean distances among the points, i.e. it has the expression given in Eq. \eqref{w} with $p=1$.\label{fig:match2d}}
\end{figure}

The EMP has been investigated by \textcite{Mezard1988} assuming the Euclidean correlation among the weights as a perturbation to the purely random case. The adopted strategy was to include, in a replica computation, only triangular correlation (i.e., the correlation among three weights), neglecting higher orders. This approach was proved successful, as numerically verified in Ref.~\cite{Houdayer1998}. Their work inspired the present contribution and will be therefore discussed more carefully below. Many results have been obtained for the EAP as well. In particular, apart from fundamental geometric properties of the solution \cite{Holroyd2010} and rigorous results on the scaling of the optimal cost \cite{Ajtai1984,Dobri1995,BoutetdeMonvel2002}, a successful ansatz for the $p=2$ case was recently proposed in Refs.~\cite{Caracciolo2014,Caracciolo2015s} for the evaluation of the average optimal cost and of the correlation functions of the solution. This ansatz was later justified through a functional approach \cite{Caracciolo2015,Sicuro2017}. In the one-dimensional case, in particular, a correspondence between the solution of the problem and a Gaussian stochastic process emerged \cite{Boniolo2012,Caracciolo2014c,Sicuro2017}.

In the present paper, we will consider the EMP on the unit hypercube in $d$ dimensions \footnote{In order to avoid trivial finite volume effects proportional to the number of points near to the surface, we will consider periodic boundary conditions, where the Euclidean distance ${\|\Bx_i-\Bx_j\|}$ is computed between the point $\Bx_i$ and the nearest of the all the images of $\Bx_j$ in the set $\{\Bx_j+\mathbf n\}_{\mathbf n}$, $\mathbf n\in\Z^d$ being a vector with integer components.}. We will improve the calculation of \textcite{Mezard1988}, going beyond the triangular approximation and including all one-loop, or polygonal, corrections to the pure mean field case. In particular, we will show that polygonal corrections can be written down, after some calculations, in a numerically manageable form. The paper is organized as follows. In Section \ref{sec:cluster} we set up a replicated formalism for the EMP, dealing with Euclidean correlations through Mayer's cluster expansion. We provide also a set of diagrammatic rules emerging from the theory, that allow us to evaluate the contribution of each diagram in the expansion. In Section \ref{sec:saddle}, we proceed performing a saddle point action approximation and, moreover, imposing a replica symmetric assumption. In Section \ref{sec:polyg} we focus our attention on a specific class of diagrams appearing in the expansion, e.g., the class of one-loop diagrams, which we call polygons: we treat the polygonal contribution using a (replicated) transfer matrix formalism. In Section \ref{sec:numeric} the asymptotic cost in the polygonal approximation is computed for different dimensions $d$ and for $p=1$. We also give some details on the spectral properties of the transfer matrix operators and highlight the presence of a null mass at zero momentum. Finally in Section \ref{sec:subsub} we will show how, from the formal structure of the polygonal series, some non-trivial finite size correction exponents in the random link case can be derived. 

\section{Cluster Expansion}\label{sec:cluster}
As usual in statistical physics' analysis of optimization problems, and following \textcite{Mezard1985,Mezard1988}, we shall associate a partition function to a given instance of the EMP. Let us assume that a set of $N$ points $\{\Bx_i\}_{i=1,\dots,N}$ is given on the $d$-dimensional unit cube $[0,1]^d$, with associated cost matrix elements $w_{ij}\coloneqq\|\Bx_i-\Bx_j\|^p$, with $p>0$. The points $\{\Bx_i\}_i$ are supposed uniformly and independently distributed in $[0,1]^d$. We define
\begin{multline}
Z_w(\beta)\coloneqq\sum_{\text{matchings }\fm}e^{-\beta N^{\frac{p}{d}} \sE_w(\fm)}
\\=\left(\prod_{i=1}^{N}\int_0^{2\pi}\frac{e^{i\lambda_i}\,d\lambda_i}{2\pi}\right)\prod_{i<j}\left[1+e^{-\beta N^{p\over d} w_{ij}-i\lambda_i-i\lambda_j}\right],
\end{multline}
where the Lagrange multipliers $\{\lambda_i\}_{i}$ enforce the constraints in Eq.~\eqref{constraints}. As discussed in Ref.~\cite{Mezard1988}, the factor $N^{\frac{p}{d}}$ is necessary in order to have an appropriate large $N$ limit for thermodynamic functions, when the average over points' positions is considered. Denoting by $\overline{\bullet}$ the expectation over the points' positions, the average free energy density of the system is given by
\begin{equation}
f(\beta)\coloneqq \lim_{N\to\infty}-\frac{1}{\beta N}\,\overline{\ln Z_w(\beta)}.
\label{def:free}
\end{equation}
It is convenient to define a rescaled average optimal cost $\hat \sE$, so that the following relations hold:
\begin{equation}
\hat \sE =
\lim_{N\to\infty} \overline{N^{\frac{p}{d}-1}\sE^*_w}
=\lim_{\beta\to\infty} f(\beta).
\end{equation}
We deal with the average over the disorder using the replica trick
\begin{equation}
f(\beta)=\lim_{N\to \infty}\lim_{n\to 0}\frac{1-\overline{Z^n_w(\beta)}}{nN\beta}.
\end{equation}
As usual we will consider an integer number $n$ of replicas during the computation, and then we will perform analytic continuation for $n\downarrow 0$.
The average replicated partition function reads
\begin{equation}
\overline{Z^n}=
\left(\prod_{a=1}^n\prod_{i=1}^{N}\int_0^{2\pi}\frac{e^{i\lambda^a_i}\,d\lambda^a_i}{2\pi}\right)\overline{\prod_{i<j}\left(1+u_{ij}\right)},
\label{z-repl}
\end{equation}
where we have introduced the quantity
\begin{equation}
u_{ij}\coloneqq\sum_{r=1}^n e^{-r\beta N^{p\over d} w_{ij}}\sum_{1\leq a_1<\dots<a_r\leq n}e^{-i\sum_{m=1}^r\left(\lambda_i^{a_m}+\lambda_j^{a_m}\right)}.
\end{equation}
In the random link matching problem the average
\begin{equation}\label{product}
\overline{\prod_{i<j}\left(1+u_{ij}\right)}
\end{equation}
is easily performed, using the fact that the joint probability distribution of the weights $\{w_{ij}\}_{ij}$ factorizes \cite{Mezard1985}. In our case, however, this is not true anymore, due to the underlying Euclidean structure. In particular, the function $u_{ij}$ depends on the vertexes $i$ and $j$ because of both the Euclidean distance $\|\Bx_i-\Bx_j\|$, and the two sets of Lagrange multipliers $\{\lambda_i^a\}_a$ and $\{\lambda_j^a\}_a$. The quantity in Eq.~\eqref{product} can be therefore represented through a diagrammatic expansion, in complete analogy with the classical cluster expansion \cite{Salpeter1958,Hansen1990,Itzykson1991} introduced by \textcite{Mayer1940}, with $u_{ij}$ playing the role of the Mayer function. In particular, applying the results of \textcite{Pulvirenti2012}, we can write
\begin{equation}
\overline{\prod_{i<j}\left(1+u_{ij}\right)}\sim \exp\left(\sum_{\substack{\fg\subseteq\fK_N\\\text{biconn.}}}\ \overline{\prod_{e\in \fg} u_e}\right).
\label{mediaemmp1}
\end{equation}
The sum on the r.h.s.~runs over all biconnected sugraphs $\fg$ of the complete graph $\fK_{N}$. A biconnected graph is a graph that remains connected after the removal of any vertex with all adjacent edges (see Fig.~\ref{fig:biconn}). Here and in the following we will denote by $(i,j)\in \fg$, or equivalently $e\in \fg$, an edge of the graph $\fg$. Moreover we will denote by $E_\fg$ and $V_\fg$, or simply $E$ and $V$, the number of edges and the number of vertexes in $\fg$ respectively.
\begin{figure}[t]
\includegraphics{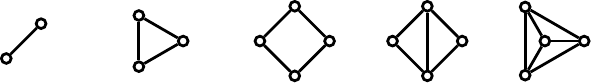}
\caption{\label{fig:biconn} Biconnected graphs up to four vertexes. The corresponding symmetry factors
$\sigma_\fg$ are $\frac{1}{2}$, $\frac{1}{6}$, $\frac{1}{8}$, $\frac{1}{4}$ and
$\frac{1}{24}$ from left to right.}
\end{figure}
In the mean field approximation only subgraphs with $E=1$ are considered. The average appearing in the arguments of the sum in Eq.~\eqref{mediaemmp1} removes the dependencies on the point positions. However, each contribution still depends on the indexes of the vertexes of the specific subgraph through the Lagrange multipliers $\{\lambda_i^a\}_{i,\,a}$. We introduce therefore a set of order parameters, symmetric under permutations of replica indexes, defined by
\begin{equation}
Q_{a_1\dots a_k}\coloneqq\frac{1}{N}\sum_{i=1}^{N}\exp\left(-i\sum_{j=1}^k\lambda_i^{a_j}\right),\quad 1\leq k\leq n,
\end{equation}
and the associated Lagrange multipliers $\hat Q_{a_1,\dots,a_k}$. In the large $N$ limit, for a given subgraph $\fg$, there are approximately $N^V\sigma_\fg$ subgraphs in $\fK_N$ isomorphic to $\fg$, $\sigma^{-1}_\fg$ being the number of automorphisms of $\fg$. With these considerations in mind, and after some simple manipulations, we can write
\begin{equation}
\sum_{\substack{\fg\subseteq\fK_N\\\text{biconn.}}}\overline{\prod_{e\in \fg} u_{e}}=
-n \beta N \sideset{}{'}\sum_{\fg} S_{\fg}[\beta,Q].
\label{psi1}
\end{equation}
The primed sum runs over all biconnected graphs with vertices labelled in $[V]$, for $2\leq V\leq N$, considered up to an automorphism. In Eq.~\eqref{psi1} we have collected a factor $-n \beta N$ for later convenience. The contribution of each graph  is given by
\begin{multline}
\label{psi2}
-n \beta S_{\fg}[\beta,Q]\coloneqq \sigma_{\fg} \sum_{\substack{\{\Ba^e\}_e}}\ \overline{\prod_{e\in \fg} e^{-\beta |\Ba^e| w_e}}^{\,\fg}\prod_{\substack{v=1}}^V Q_{\Ba(v)}\, \delta_{\Ba(v)}.
\end{multline}
Here, for each edge $e=(u,v)$, we have a sum over all non-empty subsets $\Ba^e\coloneqq\{a^e_k\}_k\subseteq[n]$, whose cardinalities are denoted by $|\Ba^e|$. We have also defined
\begin{equation}
\Ba(v)\coloneqq\bigcup_{u\in \partial v} \Ba^{(u,v)},
\end{equation}
union over the set $\partial v$ of the vertexes adjacent to $v$ in $\fg$ (see Fig.~\ref{fig:replica}). The indicator function $\delta_{\Ba(v)}$ takes value one if the incident edges of $v$ have distinct replica indexes, zero otherwise. This implies
\begin{equation}
r_e\coloneqq |\Ba(v)|=\sum_{u\in \partial v}\abs{\Ba^{(u,v)}}.
\end{equation}
Finally, the average in Eq.~\eqref{mediaemmp1} must be performed using the joint costs' distribution for a graph $\fg$
\begin{multline}
\label{rho}
\rho_\fg(\{w_e\}_{e})=\\=
\left(\prod_{u=1}^V \int_{\mathbb{R}^d} \,d^dx_u\right) \prod_{(u,v)\in \fg}\delta \left(w_{uv}-\norm{\mathbf x_u - \mathbf x_v}^p\right) \delta(\mathbf{x}_1).
\end{multline}

Let us make, now, a final remark. The strategy of Ref.~\cite{Mezard1988} was to perform the explicit computation of $\rho_{\fK_3}(w_{12}, w_{23},w_{31})$ for the triangular graph $\fg\equiv\fK_3$. Since this procedure is not easily generalizable, we can adopt a different approach. We can assign a momentum to each edge in the graph $\fg$, writing Eq.~\eqref{psi2} in the Fourier space as
\begin{multline}
\label{psi3}
-n \beta S_{\fg}[\beta,Q]=\\
=\frac{\sigma_{\fg}}{(2\pi)^{d(E-V+1)}}\prod_{e\in \fg} \left[\sum_{\substack{\Ba^e}} \int\,d^d k_e\  g_{|\Ba^e|}(k_e)\right]\\
\times \prod_{\substack{v=1}}^V \left[Q_{\Ba(v)}\,\delta_{\Ba(v)} \,\delta\left(\sum_{\substack{u\in\partial v\\\text{ingoing}}}\mathbf k_{uv}-\sum_{\substack{u\in\partial v\\\text{outgoing}}}\mathbf k_{uv}\right)\right].
\end{multline}
\begin{figure*}[t]
\includegraphics[width=1.6\columnwidth]{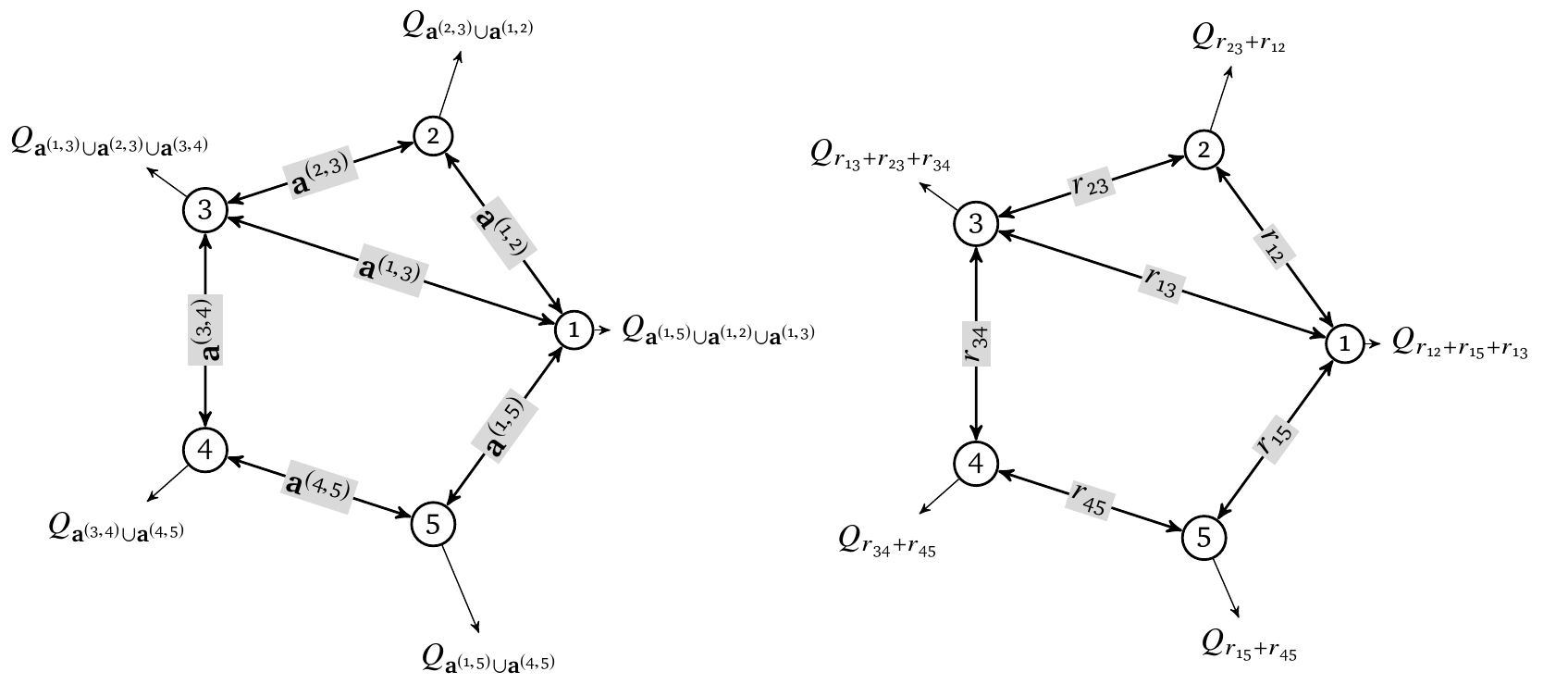}
\caption{\label{fig:replica}Pictorial representation of the construction of the replica indexes before and after the replica symmetric assumption for a biconnected graph.}
\end{figure*}
In the equation above $k_e\coloneqq\norm{\mathbf k_e}$ and
\begin{equation}
g_{r}(k)\coloneqq \Omega_{d}\int_0^\infty\, z^{d-1}e^{-r\beta  z^p}\pFq{0}{1}{-}{\frac{d}{2}}{-\frac{k^2z^2}{4}}\,d z,
\end{equation}
where ${}_0 F_1$ is a confluent hypergeometric function, defined as 
\begin{equation}
\pFq{0}{1}{-}{b}{z}\coloneqq \sum_{k=0}^\infty \frac{\Gamma(b)}{\Gamma(b+k)}\frac{z^k}{k!},
\end{equation}
and we have introduced the unit sphere's surface in $d-1$ dimensions
\begin{equation}
\Omega_d\coloneqq\frac{2\pi^\frac{d}{2}}{\Gamma\left(\frac{d}{2}\right)}.
\end{equation}
Note that a Dirac's delta function enforces the conservation of momentum on each vertex. As an additional prescription, one of the Dirac's delta has to be considered a Kronecker's delta in order to avoid an extra volume contribution. Feynman rules characterizing a generic diagram $\fg$ for the construction of Eq.~\eqref{psi3} are, at this point, given. An alternative and equivalent formulation for the momentum integration is given in Appendix \ref{app:genSg}.

\section{Replica symmetric assumption and saddle point approximation}\label{sec:saddle}
The results in the previous Section allow us to express the replicated partition function in a form that is suitable for a saddle point computation, i.e.,
\begin{equation}
\label{Act}
\overline{Z^n}\sim\left[\prod_{\Ba}\int_{-\infty}^{+\infty}\,d Q_{\Ba}\int_{-i\infty}^{+i\infty}\,\frac{d \hat Q_{\Ba}}{2\pi}\right]e^{-n \beta N S[\beta,Q,\hat Q]},
\end{equation}
where the product runs over the non-empty subsets of replica indexes $\Ba\subseteq [n]$. The action in the exponent in Eq.~\eqref{Act} has the structure
\begin{equation}\label{azionecompleta}
S[\beta,Q,\hat Q]\coloneqq S_{\text{mf}}[\beta,Q,\hat Q]+\sideset{}{'}\sum_{\fg\colon E_\fg\geq 3} S_\fg[\beta,Q].
\end{equation}
The first contribution is the mean field term, corresponding to the biconnected graph with one edge only, plus other terms deriving from the constraints imposed for the introduction of the order parameters $Q_{\Ba}$ and $\hat Q_{\Ba}$. It is given by
\begin{multline}\label{Smf}
-n\beta S_{\text{mf}}=-\sum_{\Ba} Q_{\Ba}\, \hat{Q}_{\Ba}
+\frac{1}{2} \sum_{\Ba} g_{|\Ba|}(0)\, Q_{\Ba}^2\\
+\ln\left[\prod_{a=1}^n\int_0^{2\pi}\frac{e^{i\lambda^a}\,d\lambda^a}{2\pi} \exp\left(\sum_{{\Ba}} \hat{Q}_{\Ba}\, e^{-i \sum_{l=1}^{|{\Ba}|} \lambda^{a_l}}\right)\right].
\end{multline} 

Before performing the analytic continuation for small $n$, we assume a \textit{replica symmetric ansatz}, i.e.,
\begin{equation}
Q_{\Ba}\equiv Q_{|\Ba|}\quad\text{and}\quad\hat Q_{\Ba}\equiv \hat Q_{|\Ba|}.
\end{equation}
It is convenient, in order to take the $n\downarrow 0$ and $\beta\uparrow\infty$ limits, to introduce a functional representation for the order parameters, namely
\begin{equation}\label{G}
G(x)\coloneqq \sum_{r=1}^\infty(-1)^{r-1}\frac{\hat Q_re^{\beta r x}}{r!}.
\end{equation}
The saddle point condition with respect to $\{\hat Q_r\}_r$ in the $n\to 0$ limit immediately yields
\begin{equation}
\label{saddle1}
\frac{\delta S}{\delta\hat Q_r}=0\Rightarrow Q_r=\beta\int\frac{e^{\beta rx-G(x)}}{(r-1)!}\,d x.
\end{equation}

If we restrict ourselves to the mean field approximation, in the limit $n\downarrow 0$ we can express the saddle-point mean field action as function of $G$ only (see Appendix \ref{app:meanfield} for a detailed computation)
\begin{multline}\label{smf}
S_{\text{mf}}=-\int\left(e^{-e^{\beta x}}-e^{-G(x)}\right)\,d x+\int G(x)e^{-G(x)}\,d x\\
-\frac{1}{2}\int\rho(w)e^{-G(x)-G(y)}\frac{\partial J_0\left(2e^{\beta\frac{x+y-w}{2}}\right)}{\partial x}\,dx\,dy\,dw,
\end{multline}
where $J_0(x)$ is a Bessel function of the first kind and
\begin{equation}
\rho(w) = \frac{\Omega_d}{p}\ w^{\frac{d}{p}-1} \theta(w)
\end{equation}
is the distribution of the weight appearing in the graph with $E=1$.
Taking the zero temperature limit of Eq.~\eqref{smf}, we obtain the mean field cost
\begin{multline}\label{emf}
\sE_{\text{mf}}=-\int\left(\theta(-x)-e^{-G(x)}\right)\, d x+\int G(x)e^{-G(x)}\, d x\\
+\frac{1}{2}\int\rho(x+y)\ e^{-G(x)-G(y)}\, d x\, d y
\end{multline}
The saddle point condition for $G(u)$ in the mean field approximation is then
\begin{equation}\label{EqGmf}
\frac{\delta \sE_\text{mf}}{\delta G(u)}=0\ \Rightarrow\  G(u)=\int\rho(w)\, e^{-G(w-u)}\,dw,
\end{equation}
to be used in Eq. \eqref{emf} to obtain the mean field approximation to the optimal cost. 
As anticipated, the mean field case was discussed in Ref.~\cite{Mezard1985} in the study of the random link matching problem.

If we consider, instead, the complete action, each term $S_{\fg}$ gives a correction to the mean field contribution that, in general, is of the same order of the mean field contribution itself, being the dependence of $S_\fg$ from $N$ already factorized out for large $N$, as in Eq.~\eqref{Act}. However, it has been observed by \textcite{Houdayer1998} that the contribution of the different graphs is exponentially small in the dimension $d$ of the Euclidean space.
Defining the zero-temperature limits $\sE\coloneqq\lim_{\beta\uparrow\infty}S$  and $\sE_\fg\coloneqq \lim_{\beta\uparrow\infty} S_\fg$, which can be conveniently considered as functionals of $G(u)$, saddle point extremization gives
\begin{multline}\label{saddleG}
\frac{\delta \sE}{\delta G(u)}=0\ \Rightarrow\\
G(u)=\int\rho(w)e^{-G(w-u)}\,dw-e^{G(u)}\sideset{}{'}\sum_{\fg\colon E_\fg\geq 3}\frac{\delta \sE_\fg}{\delta G(u)}.
\end{multline}
The resulting order parameter $G$ can then be used to evaluate the average optimal cost
\begin{equation}\label{exactE}
\hat \sE  = \sE_{\text{mf}} + \sideset{}{'}\sum_{\fg\colon E_\fg\geq 3} \sE_\fg.
\end{equation}
In the next Section we take a first step beyond the mean field approximation, considering the terms in the series corresponding to graphs having a single loop, and ignoring the others.

\section{One-loop contributions}\label{sec:polyg}
In this Section, we consider the one-loop terms appearing in the action in Eq.~\eqref{azionecompleta}. We denote by $\fp_E$ the one-loop graph having $E$ vertexes and $E$ edges and we will use the term \textit{polygon} for such graphs.
Polygons appear also as first finite size corrections in random link matching problem \cite{Parisi2002,ER13,RRG14} and as first corrections in certain perturbative expansions around the Bethe approximation \cite{Vontobel2011,Mori2012,Lucibello2015}. We shall denote by $S_E\coloneqq S_{\fp_E}$ the contribution of the polygon $\fp_E$ to the action in Eq.~\eqref{psi3}. The symmetry factor of a polygon $\fp_E$ is given by
\begin{equation}
\sigma_{\fp_E}=\frac{1}{2E}.
\end{equation}
Neglecting non-polygonal contributions, we thus approximate the full replicated action in Eq.~\eqref{azionecompleta} by
\begin{equation}
S_{\text{poly}}\coloneqq S_{\text{mf}}+\sum_{E=3}^\infty S_E.
\end{equation}
To explicitely compute the terms $S_E$, we can proceed in analogy with the computation performed in Ref.~\cite{Parisi2002} for the finite size corrections in the random link problem. We introduce the $(2^n-1)\times(2^n-1)$ matrix $\mathsf T(k)$, also called replicated transfer matrix, whose elements are given by
\begin{equation}
T_{\Ba\Bb}(k)\coloneqq \delta_{\Ba\cap\Bb=\emptyset}\ Q_{\abs{\Ba}+\abs{\Bb}}\sqrt{g_{\abs{\Ba}}(k)g_{\abs{\Bb}}(k)}.
\end{equation}
Here $\Ba$ and $\Bb$ are, as before, non-void elements of the power set of the replica indexes $[n]$, whose cardinality is expressed as $\abs{\Ba}$ and $\abs{\Bb}$ respectively, and $\delta_{\Ba\cap\Bb=\emptyset}$ is defined by
\begin{equation}
\delta_{\Ba\cap\Bb=\emptyset}=\begin{cases}1&\text{if $\Ba\cap\Bb=\emptyset$}\\ 0&\text{otherwise.}\end{cases}
\end{equation}
Therefore the contribution of the polygon $\fp_E$, according to Eq. \eqref{psi3} and under the replica symmetric assumption, can be written as
\begin{equation}
-n\beta S_E =\frac{1}{2 E}\frac{\Omega_d}{(2\pi)^d}\int_0^\infty k^{d-1}\tr{\mathsf T^E(k)}\,d k.
\end{equation}
To proceed further, we will diagonalize $\mathsf T(k)$ following the classical strategy of \textcite{DeAlmeida1978} and already adopted in Ref.~\cite{Parisi2002}. In fact, the next steps of our calculation, reported in Appendix \ref{app:polygons}, differ from the ones of  Ref.~\cite{Parisi2002} in the random link problem for the presence of the momentum variable $k$ only.

The matrix $\mathsf T(k)$ is invariant under permutations of the replica indexes, therefore we block diagonalize it according to the irreducible representations of the permutation group. The subspaces that are invariant under the action of the symmetry group are classified according to the number $q$ of distinguished replica indexes, in some appropriate basis spawning them (see Refs.~\cite{Monasson96,RTM14} for an application of the same procedure to disordered Ising models). Particular care has to be taken in the limits $n\downarrow 0$ followed by $\beta\uparrow \infty $. We give here only the final result, whereas the required computation is presented in the Appendix \ref{app:polygons}. 
The polygon cost functional $\sE_E$ is divided into two terms: the first one, $\sE^{(01)}$, accounting for the  contribution of the subspaces $q=0,1$, corresponds to the so-called longitudinal and anomalous sectors in spin glass literature; the second one, $\sE^{(2+)}$, accounts for all the other subspaces, $q\geq 2$, and it is non-zero for $E$ odd only. The average optimal cost functional is thus given by
\begin{subequations}
\label{eseriefull}
\begin{equation}
\label{eserie}
 \sE_\text{poly} = \sE_{\text{mf}}+\sum_{E=3}^\infty\left(\sE^{(01)}_E+\sE^{(2+)}_{E}\right).
\end{equation}
The term $\sE_{\text{mf}}$ here is given by Eq.~\eqref{emf}.
The contributions $\sE_E^{(01)}$ and $\sE_E^{(2+)}$, with $\sE_E=\sE_E^{(01)}+\sE_E^{(2+)}$ are given by
\begin{align}
\sE_E^{(01)}&\coloneqq\frac{(-1)^E\Omega_d}{2(2\pi)^d}\int_0^\infty k^{d-1}\tr{\mathsf H^{E-1}(0,k){\mathsf K}(k)}\,d k, \label{eserie01}\\
\sE_E^{(2+)}&\coloneqq\begin{cases}
\frac{\Omega_d}{E(2\pi)^d}\iint_0^\infty k^{d-1}\tr{\mathsf H^E(t,k)}\,d t\,d k&\text{$E$ odd},\\
0&\text{$E$ even.}
\end{cases}
\label{eserie2p}
\end{align}
\end{subequations}
In the equations above we have introduced the operator
\begin{subequations}
\begin{multline}
\label{operatorH}
\left[\mathsf H(t,k)\right]_{uv}\coloneqq \\
=\Omega_d\, e^{-\frac{G(u)+G(v)}{2}}\left.\frac{x^{\frac{d}{p}-1}\pFq{0}{1}{-}{\frac{d}{2}}{-\frac{k^2x^\frac{2}{p}}{4}}\theta\left(x\right)}{p}\right|_{x=u+v-2t},
\end{multline}
and the operator
\begin{multline}
\left[\mathsf K(k)\right]_{uv}\coloneqq\\
=\Omega_d\, e^{-\frac{G(u)+G(v)}{2}}\left.\frac{x^\frac{d}{p}}{d}\pFq{0}{1}{-}{\frac{d}{2}+1}{-\frac{k^2x^{\frac{2}{p}}}{4}}\theta(x)\right|_{x=u+v}.\end{multline}
\end{subequations}
As anticipated, the contribution $\sE_E^{(2+)}$  has an expression that is analogous to the finite size corrections computed in Ref.~\cite{Parisi2002} for the random link matching problem, whilst the sectors with $q=0,1$ produce a contribution $\sE_E^{(01)}$ that has no equivalent in that computation. 

The general saddle point equation for $G$ is given by Eq.~\eqref{saddleG}. However, keeping the polygonal contribution only, we can approximate Eq.~\eqref{saddleG} by
\begin{subequations}
\label{Gpolygfull}
\begin{multline}
\label{Gpolyg}
G(u)=\int\rho(w)e^{-G(w-u)}\,dw\\
-e^{G(u)}\sum_{E=3}^\infty\left[\frac{\delta\sE_E^{(01)}}{\delta G(u)}+\frac{\delta\sE_E^{(2+)}}{\delta G(u)}\right].
\end{multline}
The functional derivatives in Eq.~\eqref{Gpolyg} are given by

\begin{multline}
\frac{\delta\sE_E^{(01)}}{\delta G(u)}=\frac{(-1)^{E}\Omega_d}{2(2\pi)^d}\\\times\int_{0}^\infty k^{d-1}\sum_{m=0}^{E-1}\left[\mathsf H^{E-1-m}(0,k)\mathsf K(k)\mathsf H^m(0,k) \right ]_{uu}\,d k
\end{multline}
and similarly
\begin{equation}
\frac{\delta\sE_E^{(2+)}}{\delta G(u)}=-\frac{\Omega_d}{(2\pi)^d}\iint_{0}^\infty k^{d-1}\left[\mathsf H^{E}(t,k) \right ]_{uu}\,dk\,dt.
\end{equation}\end{subequations}
The computation of the spectra of $\mathsf H(t,k)$ and $\mathsf K(k)$ allows us to evaluate the polygonal correction both to the average optimal cost and to the saddle point solution for $G(u)$. The results of this computation will be presented in the next Section.

\section{Numerical results}\label{sec:numeric}
\begin{figure*}[ht]
\begin{subfigure}{0.45\textwidth}
\centering
eigenvalues of $\mathsf{H}(0,k)$
\includegraphics[width=\textwidth]{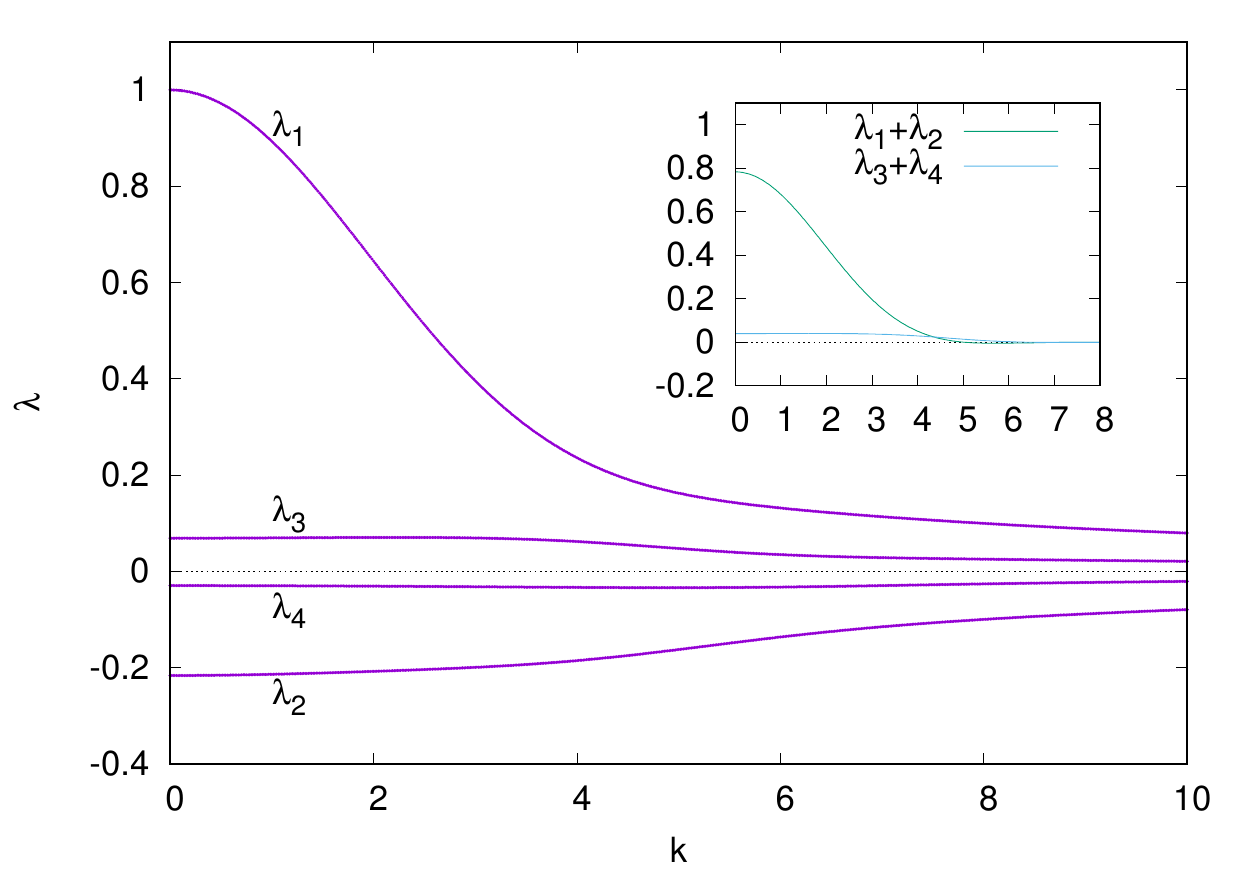}
\end{subfigure}
\begin{subfigure}{0.45\textwidth}
\centering
eigenvalues of  $\mathsf{H}(t,0)$
\includegraphics[width=\textwidth]{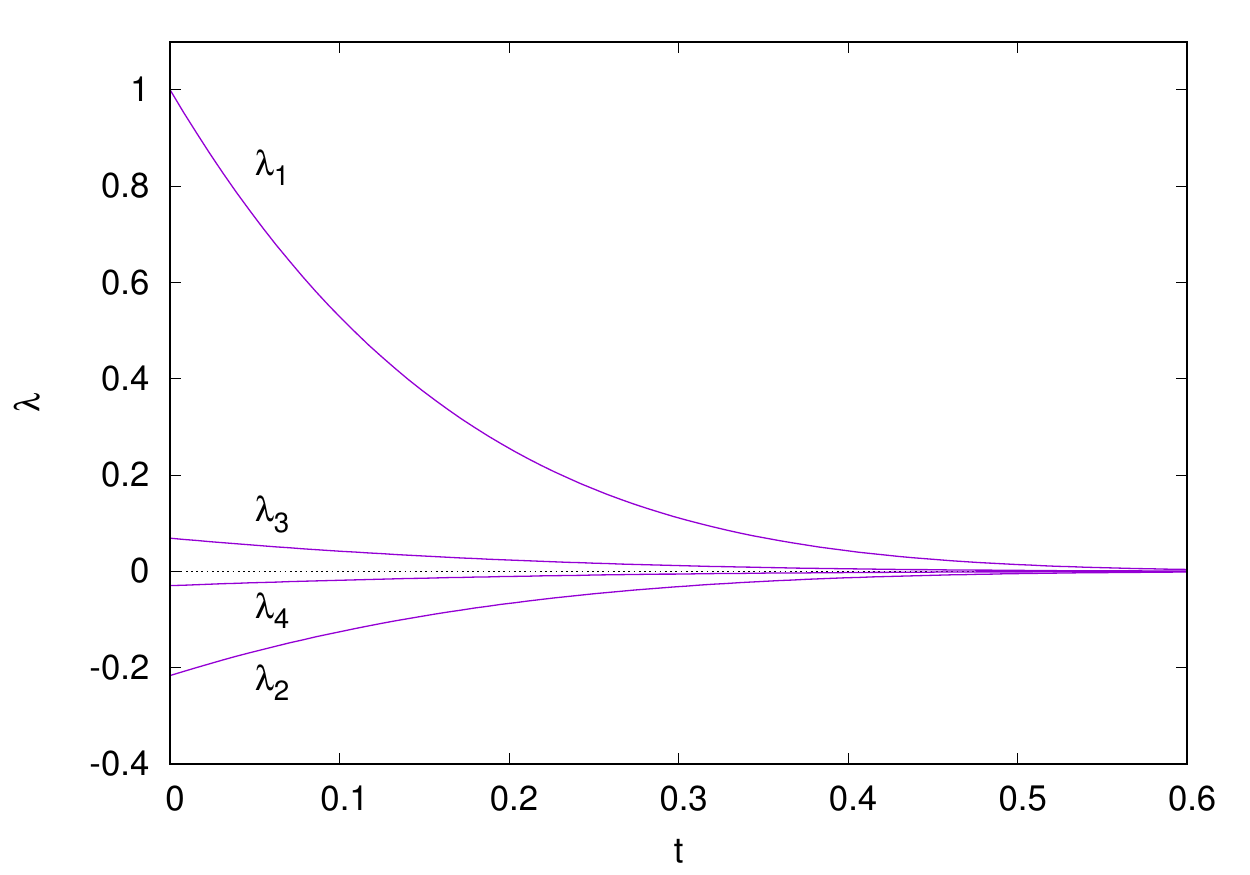}
\end{subfigure}
\caption{The four largest magnitude eigenvalues of the operator $\mathsf{H}(t,k)$ as a function of the momentum $k$ for $t=0$ (\textit{left}) and as function of $t$ for $k=0$  (\textit{right}).}   
\label{fig:eigs}
\end{figure*}
In order to compute, for a given dimension $d$ and cost exponent $p$, the polygonal approximation to the average optimal cost, we have to evaluate the cost functional $\sE_\text{poly}$, given in Eq.~\eqref{eserie}, on the solution $G_\text{poly}$ of the saddle point equation \eqref{Gpolyg}. However, a na\"ive numerical computation of the terms of the series becomes rapidly infeasible as the number of edges $E$ increases (e.g., the term  $\sE_E^{(2+)}$ in Eq.~\eqref{eserie2p} involves $E+2$ integrations). 

We adopted therefore a different strategy. We evaluated the spectrum of a discretized representation of the operator $\mathsf H(t,k)$. Typically, few of the largest eigenvalues, and the corresponding eigenvectors, are sufficient to approximate the operator within the required precision. The infinite sum over $E$ in Eq.~\eqref{eserie} and Eq.~\eqref{Gpolyg} could also eventually be taken before the integrations in $k$ and $t$. Proceeding in this way, we managed to compute efficiently $\sE_\text{poly}$ from  Eq.~\eqref{eserie} for any given $G(u)$, using the expressions
\begin{align}
\sum_{E\geq 3} \sE^{(01)}_E &= \frac{\Omega_d}{2(2\pi)^d}\int_0^\infty k^{d-1} \sum_\lambda 
\frac{-\lambda^2}{1+\lambda}\braket{\lambda|\mathsf{K}(k)|
\lambda}\,d k, \\
\sum_{E\geq 3} \sE^{(2+)}_E &= \frac{\Omega_d}{(2\pi)^d}\iint_0^\infty k^{d-1} \sum_\lambda 
(\atanh(\lambda)-\lambda)\,d t\,d k,
\label{sumE2p}
\end{align}
where the sum runs over the eigenvalues of $\mathsf{H}(0,k)$ and $\mathsf{H}(t,k)$   and we omitted the dependence of $\lambda$ from $t$ and $k$. 

On the other hand, the computation of $G_\text{poly}$ through iterations of Eq.~\eqref{Gpolyg} proved to be much harder than in the mean field case, due to some numerical instabilities that prevented the iterative procedure to reach a fixed point, even truncating the expression to the $E=3$ term. We took the alternative approach of dealing with $\sum_E \sE_E$ as a perturbation to $\sE_\text{mf}$, evaluating $\sE_\text{poly}$ in Eq.~\eqref{eserie} on $G_{\text{mf}}$, solution of Eq.~\eqref{EqGmf}. Observe that $G_{\text{mf}}(u)\, e^{-\sqrt{G_{\text{mf}}(u)}}$
is the leading eigenfunction of the operator $\mathsf H(0,0)$ with eigenvalue $\lambda_1(0,0)=1$, that is, we have a theory with a zero mass when $t=k=0$. The commutation of sum and integral leading to Eq. \eqref{sumE2p} is justified in spite of the singularity in the integrand for $t=k=0$, once one takes into account the behavior of the largest eigenvalue $\lambda_1\sim e^{-a t-b k^2}$, with $a,b>0$ and for small $t$ and $k$, as we checked numerically (see also Fig. \ref{fig:eigs}) and analytically (using perturbation theory).

We report the results of our  estimates for the (rescaled) average optimal cost in Table \ref{tab:energy}, in the case $p=1$ and using the first six eigenvalues of $\mathsf H(t,k)$. Our analytical predictions for $\hat \sE$  are compared to the numerical values $\sE_\text{num}$ obtained in Ref.~\cite{Houdayer1998}, where the authors applied an exact algorithm to random instances of the EMP and averaged over many samples. Also, for comparison, we report the values obtained for 
 \begin{equation}
 \sE_\triangle \coloneqq  \sE_\text{mf} + \sE_3,
 \end{equation}
with $\sE_\text{mf}$ and $\sE_3=\sE_3^{(01)}+\sE_3^{(2+)}$ given by Eq.~\eqref{emf} and Eqs.~\eqref{eseriefull} respectively. $\sE_\triangle$ is therefore the cost comprehensive of the triangular correlations only, as considered in Refs.~\cite{Mezard1988,Houdayer1998}. In Appendix \ref{app:triangular} we show how our expression for $\sE_3$ obtained through diagonalization in the invariant subspaces of the replica permutations group can be mapped into the expression given in Ref.~\cite{Mezard1988}.
 
We also defined $G_{\triangle}$ to be the saddle point solution for $G$ corresponding to the triangular approximation $\sE_\triangle$, and we computed it according to Eq.~(34) of Ref.~\cite{Mezard1988}. Note that a small mistake appears there in the final formulas \footnote{\label{errore} The formula for $G_\triangle$ appearing in Ref.~\cite{Mezard1988} is slightly incorrect. Indeed, the factor $2$ preceding the first integral should be changed to $2 \alpha^{\frac{d}{\nu}}$. Also, to bridge their notation with ours, one has to set $\nu=p$ and $\alpha=2^{ -\frac{p}{d}}$.}. Our results for $\sE_\triangle$ computed on $G_\triangle$ are slightly different from the ones reported in Ref.~\cite{Houdayer1998} ($\beta^{EC}$ in Table 5 of that paper). Since their numerical results were based on the analytical results in Ref.~\cite{Mezard1988}, we suspect that the discrepancy is due to the aforementioned error, that went unnoticed.

\begin{table}
\begin{tabular}{|c|c|c|c|c|c|}
\hline
 $d$ & $\sE_{\text{num}}$ & $\sE_{\text{mf}}$ in $G_{\text{mf}}$& $\sE_{\triangle}$ in $G_{\text{mf}}$   & $\sE_{\text{poly}}$ in $G_{\text{mf}}$  & $\sE_{\triangle}$ in $G_{\triangle}$ \\
\hline
1 & 0.5       & 0.4112335 &  0.33624   &   -      & 0.33623  \\
2 & 0.3104(2) & 0.3225805 &  0.29699   & 0.31376  & 0.30291  \\
3 & 0.3172(2) & 0.3268392 &  0.31255   & 0.31998  & 0.31536  \\
4 & 0.3365(3) & 0.3432274 &  0.33399   & 0.33809  & 0.33554  \\
5 & 0.3572(2) & 0.3621749 &  0.35577   & 0.35825  & 0.35669  \\
6 & 0.3777(1) & 0.3814168 &  0.37678   & 0.37838  & 0.37735  \\
\hline
\end{tabular}
\caption{Comparison of the analytical predictions for the average optimal cost for many dimension and for $p=1$. The values for $\sE_{\text{num}}$, corresponding to the average matching cost obtained by an actual matching procedure, are taken from Ref.~\cite{Houdayer1998}}
\label{tab:energy}
\end{table}

The results in Table \ref{tab:energy} show that $\sE_\text{poly}$, computed as a perturbation to $\sE_\text{mf}$, is a consistent improvement over the mean-field result in any dimension. Comparison with $\sE_{\triangle}$ in $G_{\triangle}$ is unfavourable in high dimension, though. Further investigation using the appropriate saddle point $G_\text{poly}$ are due to asses the relevance of this particular diagrammatic class in the cluster expansion.

\section{Sub-sub-leading correction in the random link problem}\label{sec:subsub}

We reconsider now the random link matching problem, that is the matching problem on $\fK_N$ with costs independently and uniformly distributed in the interval $[0,1]$. The average optimal cost $\overline{\sE_\text{RL}^*}$ has a finite asymptotic limit, computed for the first time in Ref.~\cite{Mezard1985} through the replica method, as
\begin{equation}
\lim_{N\to\infty} \overline{\sE_\text{RL}^*}\equiv\sE_\text{mf}=\frac{\pi^2}{12}.
\end{equation}
The $O(1/N)$ correction to the asymptotic cost has been obtained in Refs.~\cite{Mezard1987,Parisi2001, Ratieville2003}. In particular in Ref.~\cite{Parisi2001} it is shown that for large $N$
\begin{subequations}\label{rl-eq}
\begin{equation}
\overline{\sE_\text{RL}^*} = \sE_{\text{mf}} +\frac{\Delta{\sE}}{N} +o\bigg(\frac{1}{N}\bigg),
\label{rl-erl}
\end{equation}
with
\begin{equation}
\Delta{\sE} = \frac{\zeta(2)}{4}-\frac{\zeta(3)}{2}+\sum_{\substack{E\geq 3\\E\text{ odd}}} \frac{1}{2 E}\int_0^\infty\tr{\mathsf{H}^E(t,0)} \,d t
\label{rl-series}
\end{equation}\end{subequations}
Here we recognize the same structure of the polygonal expansion in the Euclidean case. The main differences are the absence of the momentum integration and of the $\sE^{(01)}$ term, which is equal to zero in this case \cite{Parisi2001}. The operator $\mathsf{H}$ is in fact the same we have defined in Eq.~\eqref{operatorH} for our one-loop computation in the EMP, assuming $d=p=1$. 

We will show now how the particular form of Eqs.~\eqref{rl-eq} allows us to predict the scaling with $N$ of the next order finite size correction in the random link matching problem.  
Let us start observing that, for large $E$, the integral is dominated by the region around $t=0$. 
We assume the behavior 
\begin{equation}
\tr{\mathsf{H}^E(t,0)}\sim \lambda^E(t)\sim\lambda^{-c t E},
\end{equation}
for small $t$ and large $E$, where the coefficient $c>0$ can be explicitly computed using perturbation theory. Performing the $t$ integration, we find that the coefficients of the series in Eq.~\eqref{rl-series} decay as $E^{-2}$. We can extract the sub-sub-leading scaling with $N$ of the optimal cost, which is due to counting correction in the number of loops at finite $N$, using a simple heuristic argument. A random path on the complete graph $\fK_N$ of length $\ell$ has a probability of intersecting itself in the next step of  order $\ell/N$. Therefore, for a random path of length $E$  the total probability of intersection is of order $E^2/N$ and a cross-over arises at the scale  $E\sim \sqrt{N}$.  As a consequence, at finite $N$, the sum in Eq.~\eqref{rl-series} should be opportunely regularized. Choosing an appropriate
regularizing function $f(x)$, with limits $1$ and $0$ for $x\downarrow 0$ and $x\uparrow \infty$ respectively, we have the relation
\begin{equation}
\label{rl-e32}
\sum_E \frac{1}{E^2}\ f\left(\frac{E}{\sqrt{N}}\right) \sim a + \frac{b}{\sqrt{N}},
\end{equation}
as it can be easily showed approximating the sum with an integral.
With these assumptions  the first two finite size corrections to the asymptotic cost take the form
\begin{equation}
\overline{\sE_\text{RL}^*} \sim \frac{\pi^2}{12} + \frac{e_1}{N} + \frac{e_{3/2}}{N^\frac{3}{2}}.
\end{equation}
The anomalous $\frac{3}{2}$ exponent obtained using this simple argument is indeed perfectly consistent with the numerical simulations we performed using an exact optimization algorithm \cite{Dezso2011}, see Fig.~\ref{fig:match-fullyconn}.

A refined computation of the terms appearing in the $O(1/N)$ corrections gives  $e_1=0.0674(1)$ \cite{Parisi2001,REMIDDI}, in agreement with our numerical data. From numerical fit  we then obtain the estimate $e_{3/2}=-1.24(4)$ for the coefficient of the $O(1/N^{\frac{3}{2}})$ correction.

The extension of these considerations to the polygonal contributions we computed in the Euclidean case to obtain a prediction for the exponent of the finite size correction remains an interesting open problem.

\begin{figure}
\includegraphics[width=\columnwidth]{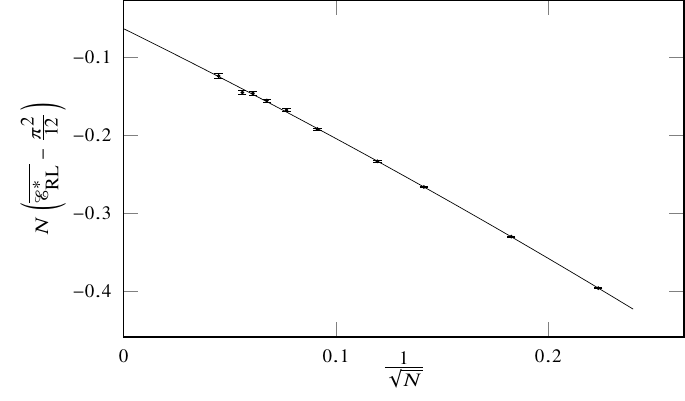}
\caption{Average optimal cost in the random link matching problem as a fuction of $N^{-1}$. Numerical data points are shown along a quadratic fit in $N^{-\frac{1}{2}}$. The fact that the data points are linearized by the chosen scaling of the axes implies a $O(N^{-\frac{3}{2}})$ finite size correction to the asymptotic average optimal cost.}
\label{fig:match-fullyconn}
\end{figure}

\section{Conclusions and perspectives}\label{sec:conclusions}
In the present work, we have discussed the random EMP on the unit hypercube in the thermodynamic limit. We have adopted the classical replica approach. It is well known \cite{Mezard1988} that Euclidean correlations among weights can be considered as corrections to a mean field contribution corresponding to the purely random case. We have shown that the Euclidean corrections can be treated in a Mayer-type expansion of biconnected diagrams, each one of them representing a different order of correlation among weights. Subsequently we restricted  our computation to the polygonal contribution in the replica symmetric hypothesis, showing that, in this case, the corrections can be properly evaluated using a transfer matrix approach. We have obtained an implicit expression for the average optimal cost in terms of the spectrum of two operators, $\mathsf K$ and $\mathsf H$. Finally, we have presented a numerical study of our results, comparing our predictions with the numerical simulations.

As specified above, in our calculation we did not evaluate non-polygonal diagrams that should be included to obtain the leading contribution to the average optimal cost. These contributions correspond to the existence of additional inner loops. An analytic treatment of these contributions would greatly improve the final theoretical predictions. Moreover, other quantities of interest related to the problem, like correlation functions, were not considered here. A restatement of the previous results in a cavity method formalism is another interesting open problem.

\section{Acknowledgments}
G.S.\@ is grateful to Sergio Caracciolo, from the University of Milan, for many, fruitful discussions. He also acknowledges the financial support of the John Templeton Foundation.
G.P. is grateful to Ettore Remiddi and Stefano Laporta for discussions on the precise estimate of some integrals and for communicating the results before publication.

\appendix

\section{Mean field Action}
\label{app:meanfield}
In the present Appendix we evaluate the mean field action presented in Section \ref{sec:saddle}. Let us start observing that, in the replica symmetric hypothesis, Eq.~\eqref{Smf} becomes
\begin{multline}
n\beta S_{\text{mf}}[\beta,Q,\hat Q]=\sum_{r\geq 1} \binom{n}{r}Q_{r}\, \hat{Q}_{r}
-\frac{1}{2} \sum_{r\geq 1} \binom{n}{r}g_{r}(0)\, Q_{r}^2\\
-\ln\left[\prod_{a=1}^n\int_0^{2\pi}\frac{e^{i\lambda^a}\,d\lambda^a}{2\pi} \exp\left(\sum_{r\geq 1}\hat{Q}_{r}\ \sum_{\mathclap{a_1<\dots<a_r}} e^{-i \sum_{l=1}^{r} \lambda^{a_l}}\right)\right].
\end{multline}
Let us now work out the $n\to 0$ limit. The result of this limit is presented already in the seminal work by \textcite{Mezard1985}. However, some intermediate, nontrivial steps are missing in their exposition and therefore we present here a more detailed derivation. We start observing that
\begin{equation}
\sum_{\mathclap{a_1<\dots<a_l}}e^{-i\sum_{j=1}^r\lambda^{a_j}}=\frac{1}{r!}\left(\sum_{a}e^{-i\lambda^a}\right)^r.
\end{equation}
It follows that
\begin{multline}
\exp\left(\sum_{r\geq 1}\hat{Q}_{r}\ \sum_{\mathclap{a_1<\dots<a_r}} e^{-i \sum_{l=1}^{r} \lambda^{a_l}}\right)\\
=\exp\left[\sum_{r\geq 1}\frac{\hat{Q}_{r}}{r!}\left(\sum_ae^{-i\lambda^a}\right)^r\right]\\
=\iint_{-\infty}^{+\infty}\exp\left[i\eta\left(x-\sum_ae^{-i\lambda^a}\right)+\sum_{r\geq 1}\frac{\hat{Q}_{r}x^r}{r!}\right]\,dx\,d\eta.
\end{multline}
The dependence on $\{\lambda^a\}_a$ factorizes and therefore we can calculate, for each value of $a$,
\begin{equation}
\int_0^{2\pi}\frac{\,d\lambda^a}{2\pi}\exp\left(i\lambda^a-i\eta e^{-i\lambda^a}\right)=i\int_\gamma \frac{e^{-i\eta z}}{z^2}\frac{dz}{2\pi}=-i\eta,
\end{equation}
where $\gamma$ is the anticlockwise oriented unit circle in the complex plane. We have
\begin{multline}
\left[\prod_{a=1}^n\int_0^{2\pi}\frac{e^{i\lambda^a}\,d\lambda^a}{2\pi}\right]\exp\left(\sum_{r\geq 1}\hat{Q}_{r}\ \sum_{\mathclap{a_1<\dots<a_r}} e^{-i \sum_{l=1}^{r} \lambda^{a_l}}\right)\\
=\iint_{-\infty}^{+\infty}(-i\eta)^n\exp\left[i\eta x+\sum_{r\geq 1}\frac{\hat{Q}_{r}x^r}{r!}\right]\,dx\,d\eta\\
=\left.\frac{d^n}{d x^n}\exp\left(\sum_{r=1}^n\frac{\hat{Q}_{r}x^r}{r!}\right)\right|_{x=0}.
\end{multline}
In the $n\to 0$ limit,
\begin{multline}
\iint_{-\infty}^{+\infty}(-i\eta)^n\exp\left[i\eta x+\sum_{r\geq 1}\frac{\hat{Q}_{r}x^r}{r!}\right]\,dx\,d\eta\\
=1+n\iint_{-\infty}^{+\infty}\ln(-i\eta)\exp\left[i\eta x+\sum_{r\geq 1}\frac{\hat{Q}_{r}x^r}{r!}\right]\,dx\,d\eta+o(n).\end{multline}
Using now the integral representation for the logarithm
\begin{equation}
\ln(x)=\int_{0}^\infty\frac{e^{-t}-e^{-xt}}{t}dt,
\end{equation}
we observe that, for a generic function $f(x)$,
\begin{multline}
\iint_{-\infty}^{+\infty}\ln(-i\eta)e^{i\eta x+f(x)}d x\,d\eta\\
=\int_0^{+\infty}\frac{dt}{t}\left[\iint_{-\infty}^{+\infty}d x\,d\eta\left(e^{-t}-e^{i\eta t}\right)e^{i\eta x+f(x)}\right]\\
=\int_0^{+\infty}\frac{e^{f(0)-t}-e^{f(-t)}}{t}dt
=\int_{-\infty}^{\infty}\left[e^{f(0)-e^y}-e^{f(-e^y)}\right]dy.
\end{multline}
Therefore, using Eq.~\eqref{G}, we have
\begin{multline}
\iint_{-\infty}^{+\infty}\ln(-i\eta)\exp\left[i\eta x+\sum_{r\geq 1}\frac{\hat{Q}_{r}x^r}{r!}\right]\,dx\,d\eta\\
=\beta\int_{-\infty}^{\infty}\left[e^{-e^{\beta y}}-e^{-G(y)}\right]dy.
\end{multline}
The other terms appearing in the mean field action can be evaluated on the saddle point using Eq.~\eqref{saddle1} and the fact that
\begin{equation}
\binom{n}{r}=\frac{(-1)^{r-1}n}{r}+o(n).
\end{equation}
In particular,
\begin{multline}
\sum_{r\geq 1} \frac{(-1)^{r-1}}{r}Q_{r}\hat Q_{r}=\beta\int e^{-G(x)} \sum_{r=1}^\infty \frac{\hat Q_r e^{\beta rx}}{r!}\,d x\\=\beta\int G(x) e^{-G(x)}\,d x
\end{multline}
and similarly
\begin{multline}
\sum_{r\geq 1} \frac{(-1)^{r-1}g_{r}(0)}{r}Q_{r}^2\\
=\beta^2\iint e^{-G(x)-G(y)}\sum_{r\geq 1} \frac{(-1)^{r-1}g_{r}(0)e^{\beta r (x+y)}}{r!(r-1)!}\,dx\,dy\\
=-\beta\iiint e^{-G(x)-G(y)}\rho(w)\frac{\partial J_0\left(2e^{\beta\frac{x+y-w}{2}}\right)}{\partial x}\,dx\,dy\,dw.
\end{multline}
Collecting all contributions, we can finally write the mean field action at finite temperature,
\begin{multline}\label{mffiniteT}
S_{\text{mf}}[\beta,Q,\hat Q]\equiv S_{\text{mf}}[\beta,G]=\int G(x) e^{-G(x)}\,d x \\
+\frac{1}{2} \iiint e^{-G(x)-G(y)}\rho(w)\frac{\partial J_0\left(2e^{\beta\frac{x+y-w}{2}}\right)}{\partial x}\,dx\,dy\,dw\\
-\int_{-\infty}^{\infty}\left[e^{-e^{\beta y}}-e^{-G(y)}\right]dy,
\end{multline}
that has the structure of Eq.~\eqref{smf}. The $\beta\to\infty$ limit of this quantity is immediately obtained using the fact that
\begin{equation}\label{limitJ}
J_0\left(2\exp\left(\frac{\beta x}{2}\right)\right)-1\xrightarrow{\beta\to\infty}-\theta(x),
\end{equation}
and therefore we have
\begin{multline}\label{smfzero}
\lim_{\beta\to\infty}S_{\text{mf}}[\beta,Q,\hat Q]=\int G(x) e^{-G(x)}\,d x \\
-\frac{1}{2} \iint e^{-G(x)-G(w-x)}\rho(w)\,dx\,dw\\
-\int_{-\infty}^{\infty}\left[\theta(-x)-e^{-G(y)}\right]dy.
\end{multline}
The mean field approximation to the optimal cost is obtained substituting  in the previous equation the mean field solution for $G(x)$, given by Eq.~\eqref{EqGmf}.

\section{Derivation of the polygonal corrections}\label{app:polygons}
To derive Eq.~\eqref{eserie}, we proceed, as anticipated, following the strategy of \textcite{DeAlmeida1978}. An eigenvector $\mathsf c=(c_\Ba)_\Ba$ of the matrix $\mathsf T$ must satisfy the equation
\begin{equation}
\sum_{\Bb} T_{\Ba\Bb}c_{\Bb}=\sum_{\Bb\colon\Ba\cap\Bb=\emptyset} Q_{\abs{\Ba}+\abs{\Bb}}\sqrt{g_{\abs{\Ba}}(k)g_{\abs{\Bb}}(k)}c_\Bb=\lambda c_\Ba.
\end{equation}
We will look for eigenvectors $\mathsf c^q$ with $q$ distinguished replicas, in the form
\begin{equation}
c_\Ba^q=\begin{cases}0&\text{if $\abs{\Ba}<q$},\\ d^i_{\abs{\Ba}}&\parbox[t]{.3\textwidth}{if $\Ba$ contains $q-i+1$ different indexes, $i=1,\dots, q+1$.}\end{cases}
\end{equation}
For $q\geq 2$, if we consider $q-1$ distinguished replicas, it can be proved \cite{Mezard1987} that the following orthogonality condition holds:
\begin{equation}
\sum_{k=0}^{q-j}\binom{k}{q-j}\binom{\abs{\Ba}-(k+j)}{n-q}d^{q+1-(k+j)}_{\abs{\Ba}}=0.
\end{equation}
The orthogonality condition provides a relation between all the different values $d^i_\abs{\Ba}$, showing that we can keep one value only, say $d^1_{\abs{\Ba}}$, as independent. Using this assumption, the eigenvalues of the original $\mathsf T(k)$ matrix can be evaluated diagonalizing the infinite dimensional matrices $\mathsf N^{(q)}(k)$ \cite{Parisi2002} whose elements, in the $n\to 0$ limit, are given by
\begin{equation}
N^{(q)}_{ab}(k)=(-1)^b\frac{\Gamma(a+b)\Gamma(b)Q_{a+b}\sqrt{g_{a}(k)g_{b}(k)}}{\Gamma(a)\Gamma(b-q+1)\Gamma(b+q)}.
\end{equation}
In particular, for $q=0$ a direct computation gives
\begin{multline}
N^{(0)}_{ab}(k)=\binom{n-a}{b}Q_{a+b}g_{b}(k)\\ \xrightarrow{n\to 0}(-1)^b\frac{\Gamma(a+b)}{\Gamma(a)b!}Q_{a+b}\sqrt{g_{a}(k)g_{b}(k)}\label{m0}
\end{multline}
whereas for $q=1$ we obtain
\begin{multline}
N^{(1)}_{ab}(k)=\binom{n-a}{b}\frac{b}{b-n}Q_{a+b}\sqrt{g_{a}(k)g_{b}(k)}\\\xrightarrow{n\to 0}N^{(0)}_{ab}(k)+\frac{n}{b}N^{(0)}_{ab}+o(n).\label{m1}
\end{multline}
Summarizing, we can write
\begin{equation}\label{SqN}
\tr{\mathsf T^E(k)}=\sum_{q=0}^\infty\left[\binom{n}{q}-\binom{n}{q-1}\right]\tr{\left(\mathsf N^{(q)}(k)\right)^E}.
\end{equation}
We distinguish now the sectors $q\geq 2$ from the sectors $q=0,1$, due to the fact that the two sets requires a different analytic treatment.
\paragraph{Sectors $q\geq 2$}
Computing the spectrum of the matrix $\mathsf N^{(q)}$ for $q\geq 2$ is equivalent to the computation of the spectrum of $\mathsf M^{(q)}(k)$, that has elements
\begin{multline}
M^{(q)}_{ab}(k)\coloneqq\\=
(-1)^{a+b}\sqrt{\frac{g_{b+q}(k)}{g_{a+q}(k)}}\frac{\Gamma(a+1)\Gamma(b+q)}{\Gamma(b+1)\Gamma(a+q)}N^{(q)}_{b+q\ a+q}(k)\\
=(-1)^{a+q}\frac{\Gamma(a+b+2q)}{\Gamma(a+2q)b!}Q_{a+b+2q}g_{b+q}(k).\end{multline}
The eigenvalue equation for $\mathsf M^{(q)}(k)$ has the form
\begin{multline}\label{Aphi}
\lambda c_{a}^{(q)}=\sum_{b=1}^\infty M^{(q)}_{ab}(k)c_{b}^{(q)}\\
=\beta(-1)^q\int\frac{(-1)^a e^{(a+q)\beta u}}{\Gamma(a+2q)}\phi^{(q)}(u;k)\,d u,
\end{multline}
where we have introduced
\begin{equation}\label{phiq}
\phi^{(q)}(u;k)\coloneqq\sum_{b=1}^\infty\frac{ e^{(b+q)\beta u-\frac{G(u)}{2}}}{b!}c_{b}^{(q)}g_{b+q}(k).
\end{equation}
Eq.~\eqref{Aphi} can be written as
\begin{equation}
\lambda\phi^{(q)}(u;k)=(-1)^{q}\int [\mathsf A^{(q)}(k)]_{uv}\phi^{(q)}(v;k)\,d v,\end{equation}
where $\mathsf A^{(q)}(k)$ is the operator
\begin{multline}
\label{Aq}[\mathsf A^{(q)}(k)]_{uv}\coloneqq\\
=\beta e^{-\frac{G(u)+G(v)}{2}+q\beta (u+v)}\sum_{a=1}^\infty\frac{(-1)^{a}e^{a\beta (u+v)}}{\Gamma(a+2q)a!}g_{a+q}(k).
\end{multline}
In the $n\to 0$ limit, from Eq.~\eqref{SqN} we have therefore
\begin{multline}\sum_{q=2}^\infty\left[\binom{n}{q}-\binom{n}{q-1}\right]\tr{\left(\mathsf N^{(q)}(k)\right)^E}\\
=\sum_{q=2}^\infty(-1)^{qE}\left[\binom{n}{q}-\binom{n}{q-1}\right]\tr{\left(\mathsf A^{(q)}(k)\right)^E}\\
\xrightarrow{n\to 0}n\sum_{q=2}^\infty(-1)^{q(E+1)}\frac{2q-1}{q(1-q)}\tr{\left(\mathsf A^{(q)}(k)\right)^E}\\
=n\sum_{q=1}^\infty\frac{4q-1}{2q(1-2q)}\tr{\left(\mathsf A^{(2q)}(k)\right)^E}\\
+(-1)^{E}n\sum_{q=1}^\infty\frac{4q+1}{2q(2q+1)}\tr{\left(\mathsf A^{(2q+1)}(k)\right)^E}.\label{sectors}
\end{multline}
Here we have used the fact that
\begin{equation}\label{limbin}
\binom{n}{q}\!-\!\binom{n}{q-1}\xrightarrow{n\to 0}\begin{cases}1&\text{if $q=0$},\\
-1+n&\text{if $q=1$},\\
n(-1)^q\frac{2q-1}{q(1-q)}&\text{if $q>1$.}\end{cases}
\end{equation}

\begin{figure}[t]
\includegraphics{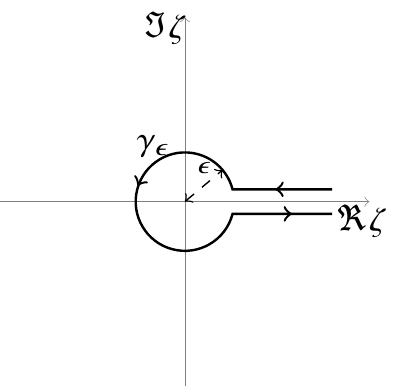}
\caption{Hankel path in the complex plane.}\label{Hankel}
\end{figure}

\paragraph{Sectors $q=0$ and $q=1$}
Let us now evaluate the contributions of the sectors $q=0$ and $q=1$. We have
\begin{multline}
\sum_{q=0}^1\left[\binom{n}{q}-\binom{n}{q-1}\right]\tr{\left(\mathsf N^{(q)}(k)\right)^E}\\
=\tr{\left(\mathsf N^{(0)}(k)\right)^E}+(n-1)\tr{\left(\mathsf N^{(1)}(k)\right)^E}+o(n).
\end{multline}
To evaluate the traces appearing in the previous expression, we define the operator $\mathsf M^{(0)}(k)$,
\begin{subequations}
\begin{multline}
M^{(0)}_{ab}(k)\coloneqq (-1)^{a+b}\sqrt{\frac{g_{b}(k)}{g_{a}(k)}}\frac{a}{b}N^{(0)}_{ba}(k)\\
=(-1)^{a}\frac{\Gamma(a+b)}{\Gamma(a)b!}Q_{a+b}g_{b}(k),\end{multline}
and the operator $\tilde{\mathsf M}^{(1)}(k)$,
\begin{multline}\tilde M^{(1)}_{ab}(k)\coloneqq(-1)^{a+b}\sqrt{\frac{g_{b}(k)}{g_{a}(k)}}\frac{a}{b}N_{ba}^{(1)}\\
=(-1)^{a}\frac{\Gamma(a+b)}{\Gamma(a)b!}Q_{a+b}g_{b}(k)\\
+n(-1)^{a}\frac{\Gamma(a+b)}{\Gamma(a+1)b!}Q_{a+b}g_{b}(k).\end{multline}
\end{subequations}
Repeating the considerations presented for the $q\geq 2$ case, we can introduce the operator $\mathsf A^{(0)}(k)$ as follows
\begin{multline}
[\mathsf A^{(0)}(k)]_{uv}=\beta e^{-\frac{G(u)+G(v)}{2}}\sum_{a=0}^\infty\frac{(-1)^ae^{a\beta(u+v)}g_a(k)}{\Gamma(a)a!}\\
=\Omega_de^{-\frac{G(u)+G(v)}{2}}\\
\times\int_0^\infty z^{d-1} \pFq{0}{1}{-}{\frac{d}{2}}{-\frac{k^2z^2}{4}}\left.\frac{\partial J_{0}\left(2e^{\beta\frac{y}{2}}\right)}{\partial y}\right|_{y=u+v-z^p}\,d z,
\end{multline}
having the same eigenvalues of $\mathsf M^{(0)}(k)$, in such a way that
\begin{equation}
\tr{\left(\mathsf N^{(0)}(k)\right)^E}=\tr{\left(\mathsf M^{(0)}(k)\right)^E}=\tr{\left(\mathsf A^{(0)}(k)\right)^E}.
\end{equation}
Similarly, we have that the eigenvalues of $\tilde{\mathsf M}^{(1)}(k)$ are obtained from
\begin{multline}
\lambda \tilde c_a=\sum_{b}\tilde{M}^{(1)}_{ab}(k)\tilde c_{b}\\
=\sum_{r'}(-1)^a\frac{\Gamma(a+b) Q_{a+b}g_b(k)}{\Gamma(a)b!}\left(1+\frac{n}{a}\right)\tilde c_{b}\\
=\int\frac{e^{a u\beta-\frac{G(u)}{2}}}{\Gamma(a)}\left(1+\frac{n}{a}\right)\tilde\phi(u;k)\,d u,
\end{multline}
where $\tilde\phi(u;k)$ is given by
\begin{equation}\label{phi0}
\tilde \phi(u;k)\coloneqq\sum_{b=1}^\infty\frac{ e^{b\beta u-\frac{G(u)}{2}}}{b!}\tilde c_{b}g_{b}(k).
\end{equation}
It is natural, therefore, to introduce the operator $\tilde{\mathsf A}^{(1)}(k)$ defined as follows
\begin{multline}
[\tilde{\mathsf A}^{(1)}(k)]_{uv}\coloneqq\\
=\beta e^{-\frac{G(u)+G(v)}{2}}\sum_{a=1}^\infty\frac{(-1)^{a}e^{a\beta (u+v)}}{\Gamma(a)a!}g_{a}(k)\left(1+\frac{n}{a}\right)\\
=[\mathsf A^{(0)}(k)]_{uv}+n [\mathsf B(k)]_{uv}.
\end{multline}
The operator $\mathsf B(k)$ introduced above is
\begin{multline}
[\mathsf B(k)]_{uv}\coloneqq\beta e^{-\frac{G(u)+G(v)}{2}}\sum_{a=1}^\infty\frac{(-1)^{a}e^{a\beta (u+v)}}{\Gamma(a+1)a!}g_{a}(k)\\
=\Omega_d\beta e^{-\frac{G(u)+G(v)}{2}}\\
\times\int_0^\infty z^{d-1} \pFq{0}{1}{-}{\frac{d}{2}}{-\frac{k^2z^2}{4}}\left[J_{0}\left(2e^{\beta\frac{u+v-z^p}{2}}\right)-1\right]\,d z.
\end{multline}
We have then, up to higher orders in $n$,
\begin{multline}
\tr{\left(\mathsf N^{(0)}(k)\right)^E}+(n-1)\tr{\left(\mathsf N^{(1)}(k)\right)^E}\\
=\tr{\left(\mathsf A^{(0)}(k)\right)^E}+(n\!-\!1)\tr{\left(\mathsf A^{(0)}(k)+n\mathsf B(k)\right)^E}\\
=n\tr{\left(\mathsf A^{(0)}(k)\right)^E}+nE\tr{\left(\mathsf A^{(0)}(k)\right)^{E-1}\mathsf B(k)}.
\end{multline}

\paragraph{Zero-temperature limit} For each one of the quantities above, we need to calculate the $\beta\to\infty$ limit, being interested in the optimal cost. Let us consider the $q\geq 2$ contribution. First, we introduce the identity
\begin{equation}\label{identita}
\sum_{r=1}^\infty\frac{(-x)^{r}}{\Gamma(r+2q)r!}=\frac{i}{2\pi}\oint_{\gamma_\epsilon} e^{-\zeta-2q\ln(-\zeta)+\frac{x}{\zeta}}\,d \zeta.
\end{equation}
The path $\gamma_\epsilon$, in the complex plane, is the Hankel path, represented in Fig.~\ref{Hankel}. This identity can be proved starting from the Hankel representation for the reciprocal gamma function \cite{abramowitz1972handbook}
\begin{equation}
\frac{1}{\Gamma(z)}=\frac{i}{2\pi}\oint_{\gamma_\epsilon}e^{-\zeta-z\ln(-\zeta)}\,d\zeta.
\end{equation}
Using Eq.~\eqref{identita}, we can rewrite Eq.~\eqref{Aq} for $q\geq 2$ as
\begin{widetext}
\begin{equation}\label{As}
[\mathsf A^{(q)}(k)]_{uv}=\frac{i\beta \Omega_d}{2\pi} e^{-\frac{G(u)+G(v)}{2}}\int_0^{+\infty}\,d w\oint_{\gamma_\epsilon}\,d\zeta \frac{w^{\frac{d}{p}-1}}{p}\pFq{0}{1}{-}{\frac{d}{2}}{-\frac{k^2w^\frac{2}{p}}{4}}\exp\left(\beta q(u+v-w)-w-2q\ln(-\zeta)+\frac{e^{\beta(u+v-w)}}{\zeta}\right).\end{equation}\end{widetext}
To compute the $\beta\to\infty$ limit, we perform a saddle point approximation, obtaining
\begin{equation}
\begin{cases}
\zeta_\text{sp}=-q,\\
w_\text{sp}=u+v-\frac{2\ln q}{\beta}.
\end{cases}\end{equation}
The saddle point has fixed position assuming that $\ln q=t\beta$ for some $t$. Taking instead $q$ fixed and $\beta\to\infty$, it is easily seen from Eq.~\eqref{Aq} that
\begin{equation}
\lim_{\beta\to\infty}[\mathsf A^{(q)}(k)]_{uv}=\begin{cases}\infty&\text{for $u+v>0$},\\0&\text{for $u+v<0$}.\end{cases}
\end{equation}
Indeed, only for $u+v-2t>0$ the saddle point is inside the range of integration. For this reason, we take $\frac{\ln q}{\beta}=t$ fixed, obtaining the limit operator $\mathsf H(t,k)$,
\begin{multline}
[\mathsf H(t,k)]_{uv}\coloneqq\lim_{\substack{\beta\to\infty,\ q\to\infty\\\beta^{-1}\ln q=t}}[\mathsf A^{(q)}(k)]_{uv}\\
\approx \frac{\Omega_d}{p} e^{-\frac{G(u)+G(v)}{2}}\left.x^{\frac{d}{p}-1}\pFq{0}{1}{-}{\frac{d}{2}}{-\frac{k^2x^\frac{2}{p}}{4}}\theta\left(x\right)\right|_{x=u+v-2t}.\label{Hkt}
\end{multline}
Observing that $\sum_{q=2}^\infty\frac{1}{\beta q}\rightarrow \int_0^{+\infty}\,d t$ the contribution to the (rescaled) average optimal cost from the $q\geq 2$ sectors is
\begin{equation}
\sE_E^{(2+)}\coloneqq\begin{cases}
\frac{\Omega_d}{E(2\pi)^d}\iint_0^\infty k^{d-1}\tr{\mathsf H^E(t,k)}\,d t\,d k&\text{$E$ odd},\\
0&\text{$E$ even.}
\end{cases}
\label{ee}
\end{equation}

For the sectors $q=0$ and $q=1$ the $\beta\to\infty$ limit can be performed quite straightforwardly. In particular, using Eq.~\eqref{limitJ}, we obtain the limit operators $\mathsf H(0,k)$,
\begin{subequations}
\begin{multline}
[\mathsf A^{(0)}(k)]_{uv}\xrightarrow{\beta\to\infty}-[\mathsf H(0,k)]_{uv}\equiv\\
-\Omega_de^{-\frac{G(u)+G(v)}{2}}\left.\frac{x^{\frac{d}{p}-1}}{p}\pFq{0}{1}{-}{\frac{d}{2}}{-\frac{k^2x^{\frac{2}{p}}}{4}}\theta(x)\right|_{x=u+v},\end{multline}
and the operator $\mathsf K(k)$,
\begin{multline}
[\mathsf B(k)]_{uv}\xrightarrow{\beta\to\infty}-\beta[\mathsf K(k)]_{uv}\coloneqq\\
=-\Omega_d\beta e^{-\frac{G(u)+G(v)}{2}}\left.\frac{x^\frac{d}{p}}{d}\pFq{0}{1}{-}{\frac{d}{2}+1}{-\frac{k^2x^{\frac{2}{p}}}{4}}\theta(x)\right|_{x=u+v}
\end{multline}
\end{subequations}
The contribution to the (rescaled) average optimal cost from the sectors $q=0$ and $q=1$ is
\begin{equation}
\sE_E^{(01)}\coloneqq (-1)^E\frac{\Omega_d}{2(2\pi)^d}\int_0^\infty k^{d-1}\tr{\mathsf H^{E-1}(0,k){\mathsf K}(k)}\,d k.
\end{equation}
Collecting the results above, Eq.~\eqref{eserie} is immediately obtained.

\section{The triangular contribution}
\label{app:triangular}
As stressed above, in Ref.~\cite{Mezard1988} only the contribution for $E=3$ was considered and discussed. For the sake of completeness, we present here the explicit computation of this contribution, starting from our formalism and specifying all details of the computation. We will show also that the expression in Eq.~\eqref{eserie} for $E=3$ is recovered from the classical result. The triangular contribution corresponds to one graph only, i.e., the triangular graph $\fK_3$. We proceed in the replica symmetric hypothesis. We observe that in this case Eq.~\eqref{psi2} becomes
\begin{multline}
-n \beta S_{3}[\beta,Q]=\\=\sum_{r_1,r_2,r_3}\frac{n!\overline{\prod_{e=1}^3 e^{-\beta r_e w_e}}^{\,\fK_3}}{6(n-r_1-r_2-r_3)!}\prod_{i=1}^3\frac{Q_{r_i+r_{i+1}}}{r_i!}\\
=\frac{n}{6}\sum_{r_1,r_2,r_3}\overline{\prod_{e=1}^3 e^{-\beta r_e w_e}}^{\,\fK_3}(r_1+r_2+r_3-1)!\\\times \prod_{i=1}^3\frac{(-1)^{r_i-1}Q_{r_i+r_{i+1}}}{r_i!}+o(n).
\end{multline}
We have used the fact that $\sigma_{\fK_3}=\frac{1}{6}$ and, moreover, there are
\begin{equation}
\frac{n!}{(n-r_1-r_2-r_3)!r_1!r_2!r_3!}
\end{equation}
ways to organize $r_1+r_2+r_3$ different replica indexes in three groups of cardinality $r_1$, $r_2$, $r_3$ respectively. Using Eq.~\eqref{saddle1} we can write the previous expression as
\begin{multline}
-\beta S_{3}[\beta,Q]=\smashoperator[l]{\sum_{r_1,r_2,r_3}}\frac{n!\overline{\prod_{e=1}^3 e^{-\beta r_e w_e}}^{\,\fK_3}}{6(n-r_1-r_2-r_3)!}\prod_{i=1}^3\frac{Q_{r_i+r_{i+1}}}{r_i!}\\
=\frac{1}{6}\left[\prod_{i=1}^3\iint\,dx_i\,d w_i G'(x_i)e^{-G(x_i)}\right]\\\times \rho_{\fK_3}(\{w_i\})K(\{\beta(x_i+x_{i+1}-w_i)\}_i),
\end{multline}
where we have introduced the function
\begin{multline}
K(\{x_i\}_i)=\\
=\sum_{r_1,r_2,r_3}(r_1+r_2+r_3-1)!\prod_{i=1}^3\frac{(-1)^{r_i-1}e^{-r_ix_i}}{(r_i+r_{i+1})!r_i!}
\end{multline}
and $\rho_{\fK_3}$ is given by Eq.~\eqref{rho2} in the form
\begin{equation}
\rho_{\fK_3}(\{w_e\})=\left[\prod_{e=1}^3\int_{\R^d}\,d^dz_e \delta\left(w_e-\norm{\mathbf z_e}^p\right)\right]\delta\left(\sum_{e=1}^3\mathbf z_e\right).
\end{equation}
Using the expression in Eq.~\eqref{mffiniteT}, we can write the action, in the triangular approximation, at finite temperature,
\begin{multline}\label{S3}
S_\text{mf}[\beta,G]+S_3[\beta,G]=\int G(x) e^{-G(x)}\,d x \\
+\frac{1}{2} \iiint e^{-G(x)-G(y)}\rho(w)\frac{\partial J_0\left(2e^{\beta\frac{x+y-w}{2}}\right)}{\partial x}\,dx\,dy\,dw\\
-\int_{-\infty}^{\infty}\left[e^{-e^{\beta y}}-e^{-G(y)}\right]dy\\-\frac{1}{6\beta}\left[\prod_{i=1}^3\iint\,dx_i\,d w_i G'(x_i)e^{-G(x_i)}\right]\\\times \rho_{\fK_3}(\{w_i\})K(\{\beta(x_i+x_{i+1}-w_i)\}_i).
\end{multline}
We have to evaluate the $\beta\to\infty$ limit. Using the identity
\begin{multline}
\frac{1}{\Gamma\left(p_1+p_2+1\right)}=\\=\frac{1}{\Gamma\left(p_1+{1\over 2}\right)\Gamma\left(p_2+{1\over 2}\right)}\int_0^1 x^{p_1-{1\over 2}}(1-x)^{p_2-{1\over 2}}\,d x,
\end{multline}
we can write
\begin{multline}
\frac{K(\{\beta x_i\}_i)}{\beta}
=-\int_{-\infty}^{+\infty} e^{- e^{-\beta w}}\left[\prod_{i=1}^{3}\int_0^1\frac{\,d u_i}{\sqrt{u_i(1-u_i)}}\right]\\\times\left[\sum_{p=1}^\infty\frac{\left(-u_i u_{i-1} e^{\beta (x_i-w)}\right)^p}{p!\Gamma^2\left(p+{1\over 2}\right)}\right]\,d w\\
\xrightarrow{\beta\to\infty}\frac{1}{\pi^3}\left[\int_0^1\frac{\,d u}{\sqrt{u(1-u)}}\right]^3\int_{0}^{+\infty}\prod_{i=1}^{3}\theta(x_i)\theta(x_i-w)\,d w\\
=\min_i(\{x_i\})\theta(x_1)\theta(x_2)\theta(x_3).
\end{multline}
To perform the last limit, we have used the fact that
\begin{equation}
\lim_{\beta\to\infty}\sum_{p=1}^\infty\frac{\left(-z e^{\beta x}\right)^p}{p!\Gamma^2\left(p+{1\over 2}\right)}=-\frac{\theta(x)}{\pi}.
\label{bessellimit}
\end{equation}
This property can be obtained applying the following
\begin{proposition}
Let $f(p)$ be an holomorphic function in the semiplane $\Re(p)>-\epsilon$, for some $\epsilon\in(0,1)$. Moreover assume that $\abs{f(p)}e^{-\pi\abs{p}}\leq M\abs{p}^{-k}$ with $k>1$ as $\Im(p)\to \pm\infty$. Then the following identity holds,
 \begin{equation}
 \lim_{x\to +\infty}\sum_{p=0}^\infty f(p)(-x)^p=0.
\label{magicformula}
\end{equation}
\end{proposition}
\begin{proof}
The series in Eq.~\eqref{magicformula} admits a representation as an integral over the Hankel path $\gamma_\epsilon$ in the complex plane, see Fig.~\ref{Hankel}, with $\epsilon\in(0,1)$:
\begin{multline}
\sum_{p=0}^\infty f(p)(-x)^p=\frac{1}{2i}\int_{\gamma_{\epsilon}} \frac{f(\zeta)x^{\zeta}}{\sin(\pi\zeta)}\,d \zeta\\
=\frac{x^{-\epsilon}}{2i}\int^{+\infty}_{-\infty} \frac{f(-\epsilon+iy)x^{iy}}{\sin\left[\pi (i y-\epsilon)\right]}\,dy,
\end{multline}
where in the second equality we have deformed the path to the vertical line $\Re(\zeta)=-\epsilon$. It follows that
\begin{equation}
\left|\sum_{p=0}^\infty f(p)(-x)^p\right| \leq x^{-\epsilon}\int^{+\infty}_{-\infty}\left|\frac{f(-\epsilon+iy)}{2\sin\left[\pi (i y-\epsilon)\right]}\right|\,dy,
\end{equation}
Given the assumptions on $f(p)$, the last integral is convergent, and the thesis follows taking the limit $x\to+\infty$.
\end{proof}
Noticeably, Eq.~\eqref{magicformula} implies
\begin{equation}
 \lim_{x\to +\infty}\sum_{p=1}^\infty f(p)(-x)^p=-f(0)
\end{equation}
from which Eq.~\eqref{bessellimit} follows immediately.

Combining the results above with the expression for the mean field action in Eq.~\eqref{smfzero}, we obtain the saddle point action in the triangular approximation and in the zero temperature limit, $\sE_\triangle\coloneqq \sE_{\text{mf}}+\sE_3$, where
\begin{multline}\label{azionetriangolare}
\sE_3=-\frac{1}{6}\left[\prod_{i=1}^3\iint\,d w_{i}\,d x_i\,G'(x_i)e^{-G(x_i)}\right]\\\times\rho_{\fK_3}(\{w_e\}_e)\min_i(\{x_{i}+x_{i+1}-w_{i}\}_{i})\prod_{i=1}^3\theta(x_{i}+x_{i+1}-w_{i}).
\end{multline}
The value of the average optimal cost can be obtained using for $G(x)$ the solution of the saddle point equation Eq.~\eqref{EqGmf}, or the solution of the saddle point equation obtained from the action $\sE_\triangle$, as showed in Ref.~\cite{Mezard1988}.

Eq.~\eqref{azionetriangolare} can be written in a different form. Indeed, expanding again the expression for $\sE_3$ using Eq.~\eqref{rho2} and the relation
\begin{multline}
\min(x_1,x_2,x_3)\theta(x_1)\theta(x_2)\theta(x_3)=\\
=\int_0^\infty  \theta(x_1-t)\theta(x_2-t)\theta(x_3-t)\,dt,
\end{multline}
we can verify that the triangular contribution can be written in terms of the operators $\mathsf H(t,k)$ and $\mathsf K(k)$ introduced in Appendix \ref{app:polygons} as
\begin{widetext}
\begin{multline}\label{triangoloHK}
\sE_3=-\frac{1}{6(2\pi)^d}\int_0^{+\infty}\,dt\int_{\R^d}\,d^d k\left[\prod_{i=1}^3\iint\,d^d z_i\,d x_i\,G'(x_i)e^{-G(x_i)+i\mathbf k\cdot\mathbf z_i}\theta(x_i+x_{i+1}-\norm{\mathbf z_i}^p-t)\right]\\
=-\frac{2\Omega_d}{3(2\pi)^d}\iint_0^{+\infty} k^{d-1}\tr{\mathsf H^3(t,k)}\,dt\,d k+\frac{\Omega_d}{(2\pi)^d}\iint_0^{+\infty} k^{d-1}\tr{\mathsf H^3(t,k)}\,dt\,d k
-\frac{\Omega_d}{2(2\pi)^d}\int_{\R^d}k^{d-1}\tr{\mathsf H^2(0,k)\mathsf K(k)}\,d^dk\\
=\frac{\Omega_d}{3(2\pi)^d}\iint_0^{+\infty} k^{d-1}\tr{\mathsf H^3(t,k)}\,dt\,d k
-\frac{\Omega_d}{2(2\pi)^d}\int_{\R^d}k^{d-1}\tr{\mathsf H^2(0,k)\mathsf K(k)}\,d^dk.
\end{multline}
\end{widetext}
The second and the third contributions in the second line derive from the fact that, given a set of three numbers $\{a_1,a_2,a_3\}$, the simple identity
\begin{multline}\textstyle
\int_{-\infty}^{\infty}\theta(t)\sum_{k=0}^2\delta'(a_{1+k}-t)\delta(a_{2+k}-t)\theta(a_{3+k}-t)\,dt\\
\textstyle+\int_{-\infty}^{\infty}\theta(t)\sum_{k=0}^2\delta(a_{1+k}-t)\delta'(a_{2+k}-t)\theta(a_{3+k}-t)\,dt=\\
\textstyle=-\int_{-\infty}^{\infty}\theta(t)\sum_{k=0}^2\delta(a_{1+k}-t)\delta(a_{2+k}-t)\delta(a_{3+k}-t)\,dt\\
\textstyle+\sum_{k=0}^2\delta(a_{1+k})\delta(a_{2+k})\theta(a_{3+k})
\end{multline}
holds. We finally have that Eq.~\eqref{triangoloHK} is exactly the contribution appearing in Eq.~\eqref{eserie} for $E=3$.

\section{Diagrammatic rules for $S_\fg$}
\label{app:genSg}
The contribution $S_\fg[\beta,Q]$ for a generic biconnected graph to the action in Eq.~\eqref{azionecompleta} can be written in a quite general form in relation to the topological structure of the graph $\fg$ itself. Let us first observe that, for a given graph $\fg$ with $V$ vertexes and $E$ edges, we can define a \textit{cycle basis} as follows \cite{Berge1973}. Every cycle in the graph can be represented in the space $\mathcal C\subseteq\{0,1\}^E$ by a vector $\fL=(\ell_e)_e$ such that $\ell_e=1$ if the edge $e$ belongs to $\fL$, $\ell_e=0$ otherwise. Remember that in a cycle, each vertex has even degree by definition, and a cycle is called \textit{circuit} if all vertexes have degree equal to two, i.e., a circuit corresponds to a ``loop'' in the nomenclature adopted in the body of the paper \footnote{In graph theory, a loop corresponds to an edge connecting a vertex to itself. This is clearly different from a circuit and from the concept of loop appearing, for example, in Section \ref{sec:polyg}. However circuits are commonly called loops in the physics literature, and we have adopted therefore this nomenclature both in the title and in the main text.}. In the introduced representation we can sum two cycles $\fL_1=(\ell_e^{(1)})_e$ and $\fL_2=(\ell_e^{(2)})_e$, in such a way that $\fL_1\oplus\fL_2=(\ell_e^{(1)}+\ell_e^{(2)}\mod 2)_e\in\mathcal C$. We say that $\mathcal L_\fg$ is a cycle basis for $\fg$ if it is a set of circuits such that every cycle in $\fg$ can be expressed as sum of circuits in $\mathcal L_\fg$, and, moreover, its cardinality $L\coloneqq \abs{\mathcal L_\fg}$ is minimal. The number $L$ is called circuit rank and, for a connected graph, it satisfies the fundamental property \cite{Berge1973}
\begin{equation}
L=E-V+1.
\end{equation}
In a planar graph a basis $\mathcal L_\fg$ can be always easily identified considering, as basis circuits, the faces of the graph. With these definitions in mind, the distribution $\rho_\fg(\{w_e\})$ in Eq.~\eqref{rho} can be written in terms of a cycle basis of the graph $\fg$ as
\begin{multline}\label{rho2}
\rho_\fg(\{w_e\})=\\
=\left[\prod_{e=1}^E\int_{\R^d}\,d^dz_e \delta\left(w_e-\norm{\mathbf z_e}^p\right)\right]\prod_{\fL\in\mathcal L_\fg}\delta\left(\sum_{e\in\fL}\mathbf z_e\right),
\end{multline}
and therefore, denoting by $r_e\coloneqq\abs{\Ba^e},$
\begin{multline}\label{mediak}
\overline{\prod_{e\in \fg} e^{-\beta r_e w_e}}^{\,\fg}=\\
=\left[\prod_{e=1}^E\int_{\R^d}\,d^dz_e \,e^{-\beta r_e\norm{\mathbf z_e}^p}\right]\prod_{\fL\in\mathcal L}\delta\left(\sum_{e\in\fL}\mathbf z_e\right)\\
=\prod_{\fL}\int_{\R^d}\frac{d^dk_\fL}{(2\pi)^d}\prod_{e=1}^E \left[\int_{\R_d}\,d^d\kappa_eg_{r_e}(\kappa_e)\delta\left(\Bk_e-\sum_{\fL\colon e\in\fL}\mathbf k_\fL\right)\right].
\end{multline}
In the equation above we have introduced the function
\begin{multline}
g_{r}(\kappa)\coloneqq\int_{\R^d}e^{i\Bk\cdot\mathbf z-\beta r\norm{z}^p}\,d^dz\\
=\Omega_d\int_0^\infty z^{d-1}e^{-\beta r z^p}\pFq{0}{1}{-}{\frac{d}{2}}{-\frac{\kappa^2z^2}{4}}\,dz.
\end{multline}
Eq.~\eqref{mediak} can be pictorially interpreted as follows. We associate to each circuit $\fL$ of our basis a ``momentum'' $\mathbf k_\fL$ and to each edge of the graph the quantity $g_{r_e}(\kappa_e)$, with the additional constraint that $\Bk_e$ is the algebraic sum of the momenta flowing in the basis circuits to which the edge $e$ belongs (see Fig.~\ref{fig:momenta}).
\begin{figure}\centering
\includegraphics[scale=0.8]{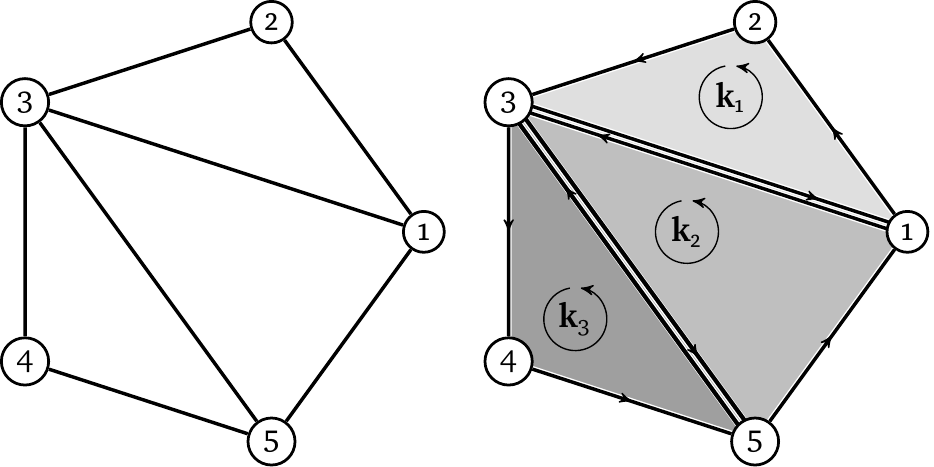}
\caption{\label{fig:momenta} A planar graph and its decomposition in basis circuits.}
\end{figure}
Inserting Eq.~\eqref{mediak} in Eq.~\eqref{psi2} we obtain a new expression depending explicitly on the topology of the graph. We can list a set of diagrammatic rules for the evaluation of $S_\fg$ at finite temperature. In particular, a momentum $\mathbf k_\fL$ must be associated to each basis circuit $\fL$; we must associate a set of replica indexes $\Ba^e$ and a quantity $g_{r_e}(\kappa_e)\delta\left(\Bk_e-\sum_{\fL\colon e\in\fL}\mathbf k_\fL\right)$ to each edge $e$, and a quantity $Q_{\Ba(v)}\delta_{\Ba(v)}$ to each vertex $v$. We must finally sum on all $\{\Ba^e\}_e$ and integrating on all momenta. Observe that the case of polygons is particularly simple, being in this case $L=1$, and therefore Eq.~\eqref{mediak} becomes
\begin{equation}
\overline{\prod_{e\in \fp_E} e^{-\beta r_e w_e}}^{\,\fp_E}=\frac{\Omega_d}{(2\pi)^d}\int \,d k\ k^{d-1}\prod_{e=1}^E g_{r_e}(k).
\end{equation}

\bibliography{biblio.bib}
\end{document}